 \documentclass[3p]{elsarticlemodified}
\usepackage{amsmath}
\usepackage{graphicx}
\usepackage{times}
\usepackage{algorithm}
\usepackage{wrapfig}

\usepackage{algorithmicx}
\usepackage[noend]{algpseudocode}
\usepackage{graphicx,caption,subcaption}
\usepackage{fullpage}
 \usepackage{amsthm}

\usepackage{epsfig}
 \usepackage{ifpdf}
 \ifpdf\fi

\usepackage{color}

\newcommand{\YM}[1]{{\bf \color{red} YM: #1}}
\newcommand{\EC}[1]{}

\newcommand{\ignore}[1]{}
\newcommand{\notinproc}[1]{#1}
\newcommand{\onlyinproc}[1]{}
\newcommand\E{\textsf{E}}

\newcommand\Et{\textsf{E}_t}  
\newcommand\Mt{\textsf{M}_t} 

\newcommand\Me{\textsf{M}_e}  
\newcommand\Ee{\textsf{E}_e}    

\newcommand\EeEt{ \Ee\Et}
\newcommand\EtEe{ \EeEt}
\newcommand\MeMt{\Me\Mt}
\newcommand\MtMe{\MeMt}
\newcommand\MeEt{\Me\Et}
\newcommand\EtMe{\Et\Me}
\newcommand\EeMt{\Ee\Mt}
\newcommand\MtEe{\Mt\Ee}

\newcommand\EEP{$\EeEt$}
\newcommand\MEP{$\MtEe$}
\newcommand\EMP{$\EeMt$}

\newcommand\WEP{$\MeEt$}
\newcommand\WMP{$\MeMt$}
\newcommand\MWP{$\EtMe$}

\newcommand\EP{$\Et$}
\newcommand\MP{$\Mt$}
\newcommand\bsigma{\boldsymbol{\sigma}}
\newcommand\bT{\text{T}}

\newcommand\MAXe{MAX$_e$\/}

\newcommand\SUMe{SUM$_e $\/}
\newcommand\Elements{V}
\newcommand\nElements{n}
\newcommand\Tests{{\cal S}}
\newcommand\nTests{m}
\newcommand\test{s}
\newcommand\nTofe{\ell}  

\newtheorem{thm}{Theorem}[section]
\newtheorem{theorem}{Theorem}[section]

\newtheorem{lemma}[thm]{Lemma}

\newtheorem{claim}[thm]{Claim}

\usepackage{setspace}
\date{}
\begin{document}
\title{Probe Scheduling for Efficient Detection of Silent Failures}

\author{Edith Cohen\corref{cor1}\fnref{fn1,fn2}} \ead{editco@microsoft.com}
\author{Avinatan Hassidim\corref{cor1}\fnref{fn3,fn4}}
\ead{avinatan@google.com}
\author{Haim Kaplan\fnref{fn2}} \ead{haimk@cs.tau.ac.il}
\author{Yishay Mansour\fnref{fn2}} \ead{mansour@cs.tau.ac.il}
\author{Danny Raz\fnref{fn5,fn6}} \ead{danny@cs.technion.ac.il}
\author{Yoav Tzur\fnref{fn3}} \ead{yoavz@google.com}

\fntext[fn1]{Microsoft Research, CA, USA}
\fntext[fn2]{Blavatnik School of Computer Science, Tel Aviv
  University,  Israel}
\fntext[fn3]{Google, Inc. Israel R\&D Center}
\fntext[fn4]{Bar-Ilan University, Israel}
\fntext[fn5]{Technion, Israel}
\fntext[fn6]{work done while at Google, Inc. Israel R\&D Center}

\ignore{
\numberofauthors{6}
\author{
\alignauthor Edith Cohen\\
       \affaddr{Microsoft Research, SVC}\\
       \affaddr{Tel Aviv University, Israel}\\
       \email{editco@microsoft.com}
\alignauthor  Avinatan Hassidim  \\
       \affaddr{Google, Inc. Israel R\&D Center}\\
       \email{avinatan@google.com}
\alignauthor  Haim Kaplan \\
       \affaddr{School of Computer Science}\\
       \affaddr{Tel Aviv University, Israel}\\
       \email{haimk@cs.tau.ac.il}
\and
\alignauthor  Yishay Mansour \\
       \affaddr{School of Computer Science}\\
       \affaddr{Tel Aviv University, Israel}\\
       \email{mansour@cs.tau.ac.il}
\alignauthor  Danny Raz \\
       \affaddr{Google, Inc. Israel R\&D Center}\\
       \email{razdan@google.com}
\alignauthor  Yoav Tzur \\
       \affaddr{Google, Inc. Israel R\&D Center}\\
       \email{yoavtz@google.com}
}
} 

%

\begin{abstract}
Most  discovery systems for silent failures work in two phases: a continuous monitoring phase that
detects presence of failures through probe packets and a localization phase that pinpoints the faulty
element(s).   We focus on the monitoring phase, where
 the goal is to balance the probing overhead with the cost associated with
longer failure detection times.

  We formulate a general model for the underlying fundamental subset-test scheduling problem.  We unify the treatment of
schedulers and cost objectives and make several contributions:
 We propose {\em Memoryless schedules} -- a natural subclass of stochastic schedules
which is simple and suitable for distributed deployment.   We show that
the optimal memoryless schedulers can be efficiently computed by convex
  programs (for SUM objectives, which minimize average detection time) or
  linear programs (for MAX objectives, which minimize worst-case
  detection time), and surprisingly perhaps,
 are guaranteed to have expected
detection  times that are not too far off the (NP hard) stochastic
optima.  We study {\em Deterministic schedules}, which provide a guaranteed bound on the 
maximum (rather than expected) cost of
undetected faults, but like general stochastic schedules, are NP hard to optimize.
We develop novel efficient
deterministic schedulers with  provable approximation ratios.
 
Finally, we conduct an experimental study, simulating our schedulers on real networks
topologies, demonstrates a significant
performance gains of the new memoryless and deterministic schedulers 
over previous approaches.
\end{abstract}


\maketitle

\section{Introduction}

Prompt detection of failures of network elements is a critical
component of maintaining a reliable network.
Silent failures,  which are not announced by the failed elements, are
particularly challenging
and can only be discovered by active monitoring.

 Failure identification systems~\cite{KompellaYGS:Infocom2007,NguyenTTD:infocom09,ZhengCao:IEEEtComp2012,ZKVM:conext2012}
typically work in two phases: First detecting presence of
a failure and then localizing it.
The rational behind this design is that
detection is an easier
problem than localization and requires light weight mechanisms that
have little impact on network performance.  Once the presence of a
failure  is confirmed, more extensive tools which may consume more
resources  are deployed for localizing the failure.  Moreover,  in some cases, it
is possible to bypass the problem, by rerouting through a different
path, quicker than the time it takes to pinpoint or correct the troubled component.

A lightweight failure detection mechanism, which relies on the
existing infrastructure, uses probe packets or {\em probes} that are sent
from certain hosts or between origin destination (OD) pairs along the existing routing
infrastructure (see Figure~\ref{network:fig}).
The elements we monitor can be physical
links~\cite{KompellaYGS:Infocom2007}, combination of components
and paths~\cite{NguyenTTD:infocom09}, or logical components of network elements like the
forwarding rules in the switches of a software-defined network \cite{ZKVM:conext2012}.
If one of the elements on the probe path fails, the probing packet will not reach the destination, and
in this case the probe has detected a failure. 
  Therefore, each probe (test)
 type can detect if at least one element in a subset of
 elements had failed.   Moreover, 
since network paths can overlap, the subsets of elements
 associated with different tests may overlap.

The goal is to design schedules which optimize the tradeoff between the probing overhead and the 
 failure detection time or more generally, the cost (or expected cost) associated with failures.
We are interested in continuous monitoring, where failures may occur at any time during the (ongoing) 
process, and we would like to detect the failure soon after it
occurs. 
Continuous testing comes in many flavors: deployment can be
centralized or distributed across the network and may require
following a fixed sequence of probes ({\em  deterministic}
schedules) or allow for randomization ({\em stochastic}  schedules).
There are also several natural objectives which we classify into two
groups. Intuitively {\em \MAXe}\ objectives aim at minimizing the maximum expected 
detection time over all  elements $e$, whereas the {\em \SUMe}\ objectives aim at
minimizing the average (or weighted sum) of the detection time. 

We illustrate
differences between these objectives through the following simple
example.
We have $n$ elements and 2 tests, one that covers a single element and another that 
covers all other $n-1$ elements.  Now, if we want to minimize the maximal expected detection time we should
send issue the tests in an alternating way (and get an expected value of $0.5$).   Any other way of 
scheduling the tests will increase the expected maximal detection
time.  On the other hand if we want to minimize the average detection
time and we assume equal failure probabilities, then it makes sense to
invoke the second test much more often
(in fact as we show later in the paper $\sqrt n$ times) than the test that covers a single element.  The two schedules described above are deterministic since they are determined by a fixed 
sequence of probes.  One can
also use a stochastic schedule in which we send each of the tests with probability $0.5$ for the 
{\em \MAXe}\ objectives, and a stochastic schedule that sends the singleton test with probability 
$1/(\sqrt n + 1)$ and the other test with probability  $\sqrt n /(\sqrt n + 1)$. 
     
We present a common framework which unifies the treatment of
stochastic and deterministic schedulers and of different objectives.
Our unified study facilitates informed design of
schedulers that are tailored to application needs.
Whilst a stronger objective, such as obtaining deterministic rather
than expected guarantees and controlling the worst-case 
rather than the average is clearly desirable, it is important to
quantify the associated costs.


%

We first present a simple and appealing sub-class of general stochastic schedules, which we call
{\em memoryless schedules}.  Memoryless schedules perform continuous testing
by invoking tests selected independently at random according to some
fixed distribution.   The stateless nature of memoryless scheduling
translates to minimum deployment
overhead and also makes them
very suitable in distributed settings,
where each type of test is initiated by a different
controller.  Going back to the example from the previous paragraph, the stochastic schedules there
are  memoryless since the distribution of the probes is fixed for each one of them.  A general stochastic
schedule for this example may be: select each of the probes with probability $0.5$, but if the long probe 
was {\em not} selected in the last 2 rounds sent the long probe.  This schedule uses the results of the 
previous steps to calculate the new probe and thus is not memoryless.
 
We show that the optimization problem of
computing the probing frequencies under which a memoryless schedule
optimizes a \SUMe\ objective can be formulated as a convex 
program
and when optimizing \MAXe\ objectives, as a linear program.  In
both cases, 
the optimal memoryless schedule can be computed efficiently.  This is
in contrast to general stochastic schedules, over  which we show that
the optima are NP-hard to compute.
Surprisingly perhaps, we also show that the natural and
efficiently optimizable memoryless schedules have expected detection times that are
guaranteed to be within a factor of two from the respective optimal stochastic
schedule of the same objective.  Moreover, detection times are
geometrically distributed, and therefore variance in detection time
is well-understood, which is not necessarily so for general stochastic schedules.
We note that our convex program formulation can be viewed as a generalization of Kleinrock's classic ``square-root law.'' 
 Kleinrock's law~\cite{kleinrock76queueing}
applies only to the special case of {\em singletons} where there is no overlap between 
the elements covered by each of the probes whereas our extension applies
to {\em subset} tests.

\ignore{
Our convex program formulation for the optimal probing frequencies
 for \SUMe\ objectives
generalizes Kleinrock's classic ``square-root law.'' 
 Kleinrock's law 
applies to the special case of  {\em singleton tests}, 
where each test  can detect the failure of a single
element (or disjoint sets of elements), and states
that the probing frequencies that minimize the \SUMe\ objective are proportional to the
square root of the weights~\cite{kleinrock76queueing}. 
Our convex program formulation applies to {\em subset tests}.
}

Another important class of schedules are {\em deterministic schedules}. Such schedules
are needed by  applications requiring hard guarantees on detection times.
Deterministic schedulers,
however, are less suitable for distributed deployment and also come
with an additional cost: the optimum of an objective on a
deterministic schedule can exceed the expectation of the same
objective over stochastic schedules. We study the inherent gap
(which we call the D2M gap) between these optima. We show that for
deterministic scheduling,  performance of \SUMe\ or \MAXe\
objectives further depends on the exact order of the quantifiers in the exact definition of
the  particular objective in the family (average or maximum). While all
variants are NP hard,
there is significant variation between attainable
approximation ratios for the different objectives.
\ignore{
: The stricter \MAXe\
and \SUMe\ objectives, requiring good performance at any point in time,
can not be approximated better than logarithmic
factors whereas the weakest, which consider average performance
over time, have factor $2$ approximations.}

Building on this, we efficiently construct  deterministic schedules
with approximation ratios that meet the analytic bounds.
Our {\em random
tree} (R-Tree)  schedulers 
derive a deterministic schedule
from the probing frequencies of a memoryless schedule,
effectively  ``derandomizing'' the schedule while attempting to
loose as little as possible on the objective in the process.
We show that when seeded, respectively,
with a \SUMe\ or \MAXe\ optimal memoryless schedule,
we obtain deterministic schedules with approximation ratio of
$O(\log \nTofe)$ for the strongest \SUMe\ objective
and ratio
$O(\log \nTofe + \log \nElements)$ for the strongest
\MAXe\ objective, where $\nElements$ is the number of elements and
$\nTofe$ is the maximum number of tests that can detect the failure
of a particular element.
We also present the Kuhn-Tucker (KT)
 scheduler which is geared to \SUMe\ objectives and adapts gracefully to changing
priorities which can be the result of changes in the network 
traffic patterns.

Finally, we evaluate the different schedulers on realistic networks
 of two different scales:
We use both a globe-spanning backbone network and a folded-Clos network, which
models a common  data center architecture.  In both cases, the elements we are testing are the
network links. For the backbone, our tests are the
set of MPLS paths and for the Clos network we use all routing paths.
We demonstrate how our suite of schedulers offers both
strong analytic guarantees, good performance,  and provides
a unified view on attainable performance with respect to different
objectives.
By relating performance of our deterministic schedulers to the
respective memoryless optima, we can see that
on many instances, our deterministic schedules are nearly optimal.
We also demonstrate how our theoretical
analysis explains observed performance and supports educated
further tuning of schedulers.

This empirical study complements the theoretical analysis in the paper and provides 
a unified general treatment of silent failures detection phase.   
Our work, by unifying the treatment of different objectives,
understanding how they relate, and developing efficient algorithms, facilitates an informed selection of
objective and algorithm that are suitable for a particular
application.

The paper is structured as follows. In Section~\ref{Model:sec} we present our model,
general stochastic and deterministic schedules, and explain the different
objectives.  Memoryless schedules are introduced in
Section~\ref{memoryless:sec}.
Deterministic scheduling is discussed in
Section~\ref{deterministic:sec}, followed by the R-Tree scheduler
in Section~\ref{treeschedules:sec} 
and Kuhn-Tucker schedulers in Section \ref{greedy:sec}.
Experimental results are presented in Section~\ref{exper:sec},
extension of the model to probabilistic tests is discussed in
Section~\ref{probtest:sec}, and
related work is discussed in Section~\ref{sec:related}.
\onlyinproc{More details, including all the missing full proofs, can be found in the Technical Report
  \cite{blackholes:arxiv2013}.}
\ignore{
Many analysis and proof details are omitted due to space limitations
and  can be found in the
 appendix which is included in the extended version of the paper
arXiv:1302.0792.}

\begin{wrapfigure}{r}{0.2\textwidth}
\centering
\ifpdf
\includegraphics[width=0.2\textwidth]{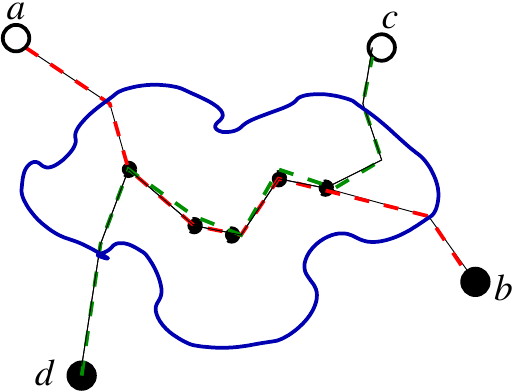}
\else
\epsfig{figure=network.eps,width=0.2\textwidth}
\fi

\caption{Network and elements covered by ab and cd origin-destination tests. \label{network:fig}}
\end{wrapfigure}

\section{Model} \label{Model:sec}

An instance of a {\em test scheduling} problem is specified by a set
$\Elements$ of elements (which can be thought of as network elements
or links) of size $\nElements$ with a  weight function
$\boldsymbol{p}$ (which can be  thought of as priority or
importance of the elements)  and a set
$\Tests$ of tests (probe paths) of size $\nTests$. For $i\in [\nTests]$, test
$i$ is specified by a subset $\test_i \subset \Elements$ of elements.
The failure of an element $e$ can be detected by probing $i$ if and only if
$e\in \test_i$, that is, if and only if test $i$ contains the failed
element. (This
can be extended to the case where failures are detected with some positive probability.)
We use $\nTofe_e$ to indicate the number of tests which
include element $e$ and $\nTofe \equiv \max_e \nTofe_e$.

Continuous testing is specified by a {\em schedule} which generates
an  infinite sequence $\sigma = \sigma_1,\sigma_2,\ldots$ of tests.
The schedule can be {\em deterministic} or {\em stochastic}, in
which case, the probability distribution of the tests at time $t$
depends on the actual tests preformed prior to time $t$.
 We also introduce {\em memoryless} schedules, which are a special
subclass of stochastic schedules, in which the probability
distribution of the tests is fixed over time. When the schedule is
stochastic we use $\bsigma$ to denote the schedule itself and
$\sigma$ to denote a particular sequence that the schedule can
generate.


\subsection{Objectives}

 Objectives for a testing schedule aim to minimize a certain function of
the number of tests invoked until a failure is detected. (We
essentially measure time passed until the failure is detected by the
``number of probes'' required to discover it. {If the probing rate
is fixed this is indeed the time}.) Several different natural
objectives had been considered in the literature.  Here we consider
all these objectives through a unified treatment which allows us to
understand how they relate to each other and how they can be
computed or approximated.

The {\em detection time} T$_{\boldsymbol{\sigma}}(e,t)$  for element $e$
at time $t$ by a schedule $\boldsymbol{\sigma}$ is the
expected time to detect a failure of element $e$ that occurs at  time
$t$.  If the schedule is deterministic, then
$\text{T}_{\boldsymbol{\sigma}}(e,t)=\min_{h\geq 0}  e\in \test_{\sigma_{h+t}}$.  If the
schedule is stochastic, we take the expectation over sequences
$$\text{T}_{\boldsymbol{\sigma}}(e,t)=\E_{\boldsymbol{\sigma}}[\min_{h\geq
  0}  e\in \test_{\sigma_{h+t}}]\ .$$
Note that the probability of any prefix is well defined for general stochastic
schedules.  Therefore $\text{T}_{\boldsymbol{\sigma}}(e,t)$, if finite,
is well defined.

  We classify natural objectives as
{\em \MAXe}, when aiming to minimize the maximum detection time over
elements, where the detection time of each element is multiplied by
its weight, or as
 {\em \SUMe} when aiming to minimize
a weighted sum  over elements of their detection times. Both types
of objectives are defined with respect to a weight function
$\boldsymbol{p}$ over elements. Objectives in each family differ by
the way they quantify over time: For example one {\em \MAXe}
objective is to minimize the maximum detection time of an edge over
all times, and a different {\em \MAXe} objective would be to
minimize the average over times of the maximum detection time of an
edge in each time. Formal definitions follow below.

The weighting, or priorities of different elements,  can capture the
relative criticality of the element
which in turn, can be set according to the
volume or quality of service level of the traffic they handle.
With the \SUMe\
objectives, the weights can also correspond to estimated  probability
that elements fail,  in which case  the weighted objective capture the
expected detection time after a failure, or to the product of failure
probability of the element and cost of failure of this element, in
which case the weighted objective is the expected cost of a failure.
  With the \MAXe\ objectives
we can use $p_e\equiv 1/\tau_e$, where
$\tau_e$ is the minimum desired detection time  for a failure of
element $e$,
or the cost of a unit of downtime of element $e$. We  then
 aim  to minimize the maximum cost of a failing element.
 In the sequel, unless otherwise mentioned,
 we assume that weights are scaled so that with \SUMe,   $\sum_e p_e
 =1$, and with \MAXe, $\max_e p_e=1$.

To streamline the definitions and treatment of the different
\MAXe\ and \SUMe\ objectives
we define the operators $\Me$ and $\Ee$, which perform
  weighted maximum or average over elements, and
$\Mt$ and $\Et$, which perform maximum or average over time.
More precisely, for a function $g$ of time or a function $f$ over
elements:
\begin{align*}
\Mt[g] &= \sup_{\tau \geq 1} g(\tau) & \quad & 
\Et[g] =\lim_{h \rightarrow \infty} \frac{\sum_{t=1}^{h}g(\tau)}{h} \\
\Me[f] &= \max_e p_e f(e) & \quad   &
\Ee[f] = \sum_e p_e  f(e)
\end{align*}
An application of the operator $\Et$ requires that the limit exists
and an application of the operator $\Mt$ requires that $g(\tau)$ is
bounded.

 When the operators are  applied to the function $\text{T}_{\boldsymbol{\sigma}}(e,t)$, we
use the shorthand $\Mt[e|\boldsymbol{\sigma}] \equiv
\Mt[\text{T}_{\boldsymbol{\sigma}}(e,t)]$,
$\Et[e| \boldsymbol{\sigma}] \equiv \Et[\text{T}_{\boldsymbol{\sigma}}(e,t)]$,
$\Me[t|\boldsymbol{\sigma}] \equiv \Me[\text{T}_{\boldsymbol{\sigma}}(e,t)]$,  $\Ee[t | \boldsymbol{\sigma}] \equiv \Ee[\text{T}_{\boldsymbol{\sigma}} (e,t)]$.
For a particular element $e$, $\Mt[e| \boldsymbol{\sigma}]$
 is  the maximum over time
$t$ of the expected (over sequences) number of probes needed to
detect a failure of $e$ that occurred in time $t$, and \EP$[e|
\boldsymbol{\sigma}]$ is the limit of the average over time $t$ of
the  expected number of probes needed to detect a failure of $e$
that occurred in time $t$.
 For a particular time $t$,
$\Me[t|\boldsymbol{\sigma}]$ is the weighted maximum over the
elements of the expected detection time of a failure at $t$, and
$\Ee[t|\boldsymbol{\sigma}]$ is
 the weighted sum  over the elements of their expected
detection times at $t$.
We consider all objectives that we can obtain from combinations of
these operators.  The operator pairs $\Me$ and $\Mt$ (maximum over time or
over elements) and $\Ee$ and $\Et$ (average of expectation) commute,
but other pairs do not, and we obtain six natural objectives,
three \MAXe\ and three \SUMe.

\smallskip
\noindent {\bf \MAXe\ objectives: } The three \MAXe\ objectives are
\begin{itemize}
\item
$\Me[\Mt[e|\boldsymbol{\sigma}]]$,
 the weighted
maximum over elements of the maximum over time of the detection
time.
\item $\Me[\Et[e|\boldsymbol{\sigma}]]$, the weighted maximum
over elements of the average over time of the detection time.
\item
$\Et[\Me[t|\boldsymbol{\sigma}]]$, the  average over time of the
maximum detection time of an  element at that time.
\end{itemize}

We shorten notation as follows.
\begin{align}
\text{\WMP}[\bsigma] &= \Me[ \Mt[e|
\bsigma]] \equiv \sup_{e,t}  p_e \text{T}_{\bsigma}(e,t) \nonumber \\
\text{\WEP}[\bsigma] &=  \Me [\Et[e|
\bsigma]] \equiv \max_e p_e \text{\EP}[e|\bsigma]\  \nonumber \\
 \label{eq:limit}
\text{\MWP}[\bsigma] &= \Et [\Me[t|\bsigma]] \equiv \lim_{h \rightarrow \infty}
\frac{1}{h}\sum_{t=1}^{h} \max_e  p_e \text{T}_{\bsigma}(e,t) \ . 
\end{align}

\smallskip
\noindent {\bf \SUMe\ objectives: } The three \SUMe\ objectives are
\begin{itemize}
\item
 $\Ee[\Mt[e|\boldsymbol{\sigma}]]$,
the weighted sum over elements $e$ of the maximum over time $t$ of
the detection time.
\item
 $\Mt[\Ee[t|\boldsymbol{\sigma}]]$,  the maximum over time of the weighted sum over
$e$  of the detection time.
\item
$\Ee[\Et[e|\boldsymbol{\sigma}]]$, the weighted sum over elements of
the average over time of the detection time.
\end{itemize}

We shorten notation as follows.
\begin{align*}
\text{\EMP}[\bsigma] &= \Ee [ \Mt[e|\bsigma]] = \sum_e p_e \text{\MP}[e|\bsigma] \\
\text{\MEP}[\bsigma] &= \Mt [ \Ee[t|\bsigma]] = \sup_t \sum_e p_e \text{T}_{\bsigma}(e,t) = \sup_t \Ee[t|\bsigma] \\
\text{\EEP}[\bsigma] &= \Ee [  \Et[e|\bsigma]] = \sum_e p_e
\text{\EP}[e | \bsigma]
\end{align*}


When the schedule ${\boldsymbol{\sigma}}$  is clear from context, we
omit the reference to it in the notation. There are clearly schedules,
deterministic or stochastic, over which our objectives are not
defined.
  The \WMP, \EMP, and \MEP\ are defined when $\Mt[e]$ is defined for all
elements $e$ and the \WEP\ and \EEP\ are defined  when $\Et[e]$ is
defined for all elements $e$.
 The \MWP\ requires that the limit in Equation (\ref{eq:limit}) exists.
Formally, we define a schedule to be {\em valid} if
for all elements $e$,
$\Mt[e]$ and $\Et[e]$ are well defined,
 and for all tests $i$
the relative frequency of probing $i$ converges, that is, the limit
$\lim_{h\rightarrow \infty} \frac{\sum_{t=1}^{h}
  \Pr[\sigma_t=i]}{h}$ exists.\footnote{Deterministic schedules that are cyclic or
    stochastic schedules with finite memory are always valid, but
    general sequences may not be.}
Henceforth we limit our attention only to valid
  schedules, which for brevity we will keep calling
schedules.

\subsection{Relating and optimizing objectives}

The following lemma specifies the basic relation between the objectives.
Its proof is straightforward.
\begin{lemma} \label{basicrel}
For any schedule $\boldsymbol{\sigma}$,
\begin{align}
\text{\SUMe :} \quad\quad & \text{\EMP}[\boldsymbol{\sigma}] \geq \text{\MEP}[\boldsymbol{\sigma}]  \geq  \text{\EEP}[\boldsymbol{\sigma}] \label{ineqW}\\
\text{\MAXe :} \quad\quad  & \text{\WMP}[\boldsymbol{\sigma}] \geq \text{\MWP}[\boldsymbol{\sigma}] \geq \text{\WEP}[\boldsymbol{\sigma}] \label{ineqPE}
\end{align}
\end{lemma}

For any objective we want to find schedules that minimize it. We
denote the infimum of the objective over  deterministic schedules by
the prefix opt$_D$, over memoryless schedules by opt$_M$, and over
stochastic schedules by opt. For example for the objective $\MeEt$,
$\text{opt$_D$-}\MeEt$ is the  infimum  $\MeEt$ over deterministic
schedules. Since memoryless and deterministic schedules are a subset
of stochastic schedules, the deterministic or the memoryless optima
are always at least the stochastic optimum: For any objective
$\text{opt}_D \geq \text{opt}$ and  $\text{opt}_M \geq \text{opt}$.

Relations
 \eqref{ineqW} and \eqref{ineqPE} clearly hold with respect to
the
 deterministic, memoryless, or stochastic
optima of each objective.
Lemma \ref{allstocequal:lemma}
shows that for stochastic schedules,
the three optima of the
objectives within each category (\SUMe\ or \MAXe) are in fact equal.

 \begin{lemma}  \label{allstocequal:lemma}
\begin{align}
\text{opt-}\EeMt  =  \text{opt-}\MtEe  =  \text{opt-}\EeEt  &  \label{soptwequal}\\
 \text{  opt-}\MeMt  =  \text{opt-}\EtMe  =
\text{opt-}\MeEt &  \ .   \label{soptpequal}
\end{align}
  \end{lemma}
\begin{proof}
The complete proof is provided in
  \onlyinproc{\cite{blackholes:arxiv2013}}\notinproc{\ignore{Appendix~}\ref{allstocequal:proof}.}
Proof sketch:
For a stochastic
schedule $\bsigma$ and a number $N$, we define a ``cyclic'' schedule
$\bsigma_N$ which repeats a prefix of $\bsigma$ of length $N$. We show that for a
sufficiently large $N$, for any item $e$, $ \Et[e| \bsigma_N] \leq
(1+\epsilon)  \Et[e| \bsigma]$.
We  randomize the start time of $\bsigma_N$ to obtain a schedule 
for which $\bT(e,t)$  is the same for all times $t$ and equals  $\Et[e|
 \bsigma_N]$.
Then \eqref{soptwequal}  follows by applying this construction to
$\text{opt-}\EeEt$ and \eqref{soptpequal} follows by applying it to
$\text{opt-}\MeEt$.
\end{proof}

We  denote by opt-\SUMe\ and opt-\MAXe\  the stochastic
optima of all three  \SUMe\ or \MAXe\ objectives:
\begin{align*}
\text{opt-\SUMe}  \equiv \text{opt-}\EeMt & = \text{opt-}\MtEe =
\text{opt-}\EeEt & \quad \mbox{and} \quad
\text{opt-\MAXe} \equiv  \text{  opt-}\MeMt & = \text{opt-}\EtMe =
\text{opt-}\MeEt   \ .
\end{align*}

We show that optimizing any of our \SUMe\ or \MAXe\ objectives is NP hard, the proof is
based on a reduction to exact cover by sets of size 3 (X3C) and is provided in \notinproc{\ignore{Appendix}
  \ref{nphardproof:sec}.}\onlyinproc{\cite{blackholes:arxiv2013}.}

\begin{lemma} \label{NPhard}
Computing the optimal schedules for opt-\SUMe and opt-\MAXe is
NP hard.
\ignore{
Computing any one of the following optima is NP hard:
opt$_D$-\EEP, opt$_D$-\MEP, opt$_D$-\EMP,  opt-\SUMe,
 opt$_D$-\WMP,  opt$_D$-\WEP,  opt$_D$-\MWP,  and opt-\MAXe.
}
\end{lemma}

\EC{ What about hardness of approximability ?  The X3C reduction is what
we used in previous work for the one-time problem.   The Feige paper
showed a stronger thing that any approximation less than $4$ is not
polynomial (for that problem).  What can we show here ?  Since we have
factor 2 approximation (memoryless) and also constant approximation
for \EEP\
(through tree schedules for deterministic) would be nice to show some
inapproximability there.  Would also be nice to show something even
stronger for \MAXe\ and worst-case time over \SUMe.}

\section{Memoryless schedules} \label{memoryless:sec}
 {\em Memoryless schedules} are particularly simple stochastic schedules
 specified by
a probability distribution   $\boldsymbol{q}$ on the tests.
 At each time, independently of history,  we draw a test $i\in[m]$
at random according to $\boldsymbol{q}$ ($i\in [m]$ is selected with probability $q_i$) and
probe $i$, where the notation $[m]=\{1,\ldots,m\}$ is the set of
integers from $1$ to $m$. It is easy to see that in memoryless schedules, detection times
are distributed geometrically.
  We show that
memoryless schedules
perform nearly as well, in terms of expected detection time, as
general stochastic schedules.  For notational convenience, we use the
distribution $\boldsymbol{q}$ to denote also
the  memoryless schedule itself.

We first show that all \SUMe\ objectives and all \MAXe\
objectives are equivalent on any memoryless schedule.
\begin{lemma}  \label{memopts:lemma}
For any memoryless schedule $\boldsymbol{q}$,
\begin{align*}
\text{\EMP}[\boldsymbol{q}]  &=\text{\MEP}[\boldsymbol{q}] = \text{\EEP}[\boldsymbol{q}] = \sum_e
\frac{p_e}{Q_e}\ \equiv \text{\SUMe}[\boldsymbol{q}] \\
\text{\WMP}[\boldsymbol{q}] &=\text{\MWP}[\boldsymbol{q}] = \text{\WEP}[\boldsymbol{q}] = \max_e
\frac{p_e}{Q_e} \equiv \text{\MAXe}[\boldsymbol{q}] \ ,
\end{align*}
where $Q_e=\sum_{i|e\in \test_i} q_i$.
\end{lemma}
\begin{proof}
  The detection time of a failure of $e$ via a memoryless schedule
is a geometric random variable  with parameter  $Q_e$.  In particular, for each element $e$,
the distribution $T(e,t)$ are identical for all $t$ and its
expectation, $1/Q_e$,  is equal to $\Mt[e]$
and \EP$[e]$.
 From linearity of expectation, the \EEP, \MEP, and \EMP\ are all equal to
 $\sum_e p_e/Q_e$.
Similarly, \WMP, \WEP, and \MWP\ are all equal to $\max_e \frac{p_e}{Q_e}$.
\end{proof}

We use the  notation opt$_M$-\SUMe\ and opt$_M$-\MAXe\ for the
memoryless optima. That is
\begin{align*}
\text{opt$_M$-\SUMe} &= \min_{\boldsymbol{q}} \text{\SUMe}[\boldsymbol{q}] \\
\text{opt$_M$-\MAXe} &= \min_{\boldsymbol{q}}  \text{\MAXe}[\boldsymbol{q}]\ .
\end{align*}

\subsection{Memoryless Optima}


 We show that the memoryless
optima with respect to both the \SUMe\ and \MAXe\ objectives can be efficiently
computed.
This is in contrast to deterministic and stochastic optima, which are
NP hard.

\begin{figure}[h]
\centering
\fbox{
\begin{subfigure}[b]{0.33\textwidth}
\begin{align}
\text{minimize}
\sum_e \frac{p_e}{\sum_{i|e\in \test_i} q_i}  \label{basicconvex}\\
\forall i,\  q_i\geq 0 \nonumber\\
\sum_i q_i =1 \nonumber
\end{align}
\caption{Convex program for \SUMe \label{basicconvex:fig}}
\end{subfigure}
}
\fbox{
\begin{subfigure}[b]{0.33\textwidth}
\begin{align}
\text{maximize}\, \ z   \label{basicLP} \\
\forall e,\  \frac{1}{p_e} \sum_{i|e\in \test_i} q_i \geq z \nonumber \\
\forall i,\  q_i\geq 0 \nonumber\\
\sum_i q_i =1 \nonumber
\end{align}
\caption{LP for \MAXe.}  \label{basicLP:fig}
\end{subfigure}
}
\caption{Computing \SUMe\ and \MAXe\ optimal memoryless schedules. \label{optmem:fig}}
\end{figure}

\begin{theorem}
The optimal memoryless schedule for \SUMe\ objectives, that is, the distribution
$\boldsymbol{q}$ such that
$\text{\SUMe}[\boldsymbol{q}]=\text{opt$_M$-\SUMe}$
is the solution of the convex program \eqref{basicconvex} (Figure \ref{optmem:fig}).
\end{theorem}

  The optimal memoryless schedules with respect to the \MAXe\
  objectives can be computed using an LP.
\begin{theorem}
The optimal memoryless schedule for \MAXe, that is, the distribution
$\boldsymbol{q}$ which satisfies
$\text{\MAXe}[\boldsymbol{q}]=\text{opt$_M$-\MAXe}$
is the solution of
 the  LP \eqref{basicLP} (Figure \ref{optmem:fig}).
\end{theorem}

\smallskip
\noindent
 {\bf Singletons instances:}  When each test is for a single element,
 the optimal solution of the convex program \eqref{basicconvex}  has
the frequencies of each element
proportional to the square root of $p_e$
 \cite{kleinrock76queueing}, that is, $q_e
=\sqrt{p_e}/\sum_e \sqrt{p_e}$.
The \SUMe\ optimum for an instance with weighting $\boldsymbol{p}$ is
\begin{equation} \label{singleEMP}
\text{opt$_M$-\SUMe}(\boldsymbol{p})=\sum_e \frac{p_e}{q_e} = \sum_e \sqrt{p_e}\sum_{i} \sqrt{p_i}= (\sum_i \sqrt{p_i})^2\ .
\end{equation}
In contrast, the solution of the LP \eqref{basicLP}  has optimal probing
frequencies $q_e$ proportional to $p_e$, that is,
$q_e = p_e/\sum_e p_e$ and the \MAXe\ optimum is
$\text{opt$_M$-\MAXe}(\boldsymbol{p})=\max_e \frac{p_e}{q_e} = \sum_e p_e$.

\subsection{Memoryless versus Stochastic}

 For both \SUMe\ and \MAXe\ objectives,  the optimum on memoryless schedules is
within a factor of 2 of the optimum over  general stochastic
schedules.
\begin{theorem}  \label{memvssto:thm}
\begin{align}
\text{opt-\SUMe} &\leq \text{opt$_M$-\SUMe}  \leq 2  \text{opt-\SUMe}  \label{georel} \\
\text{opt-\MAXe} &\leq \text{opt$_M$-\MAXe} \leq  2 \text{opt-\MAXe}
\end{align}
\end{theorem}
\begin{proof}
The left hand side inequalities follow from memoryless schedules being
a special case of stochastic schedules.  To establish the right hand
side inequalities,
consider a stochastic schedule
and let $q_i$ be (the limit of) the relative frequency of test $i$
(Recall that we only consider valid schedules, where the limit exists).
We have
 $$\text{\MP}[e] \geq \text{\EP}[e] \geq  \frac{p_e}{2\sum_{i|e\in \test_i} q_i}$$
 Therefore, the average over elements
 $\sum_e \frac{p_e}{2\sum_{i|e\in \test_i} q_i}$
must be at least half the optimum of \eqref{basicconvex} and the
maximum over elements $\max_e \frac{p_e}{2\sum_{i|e\in \test_i} q_i}$
must be at least half the optimum of \eqref{basicLP}.
\end{proof}

The following example shows that Theorem~\ref{memvssto:thm} is tight in
that the ``2'' factors  are realizable.  That is, there
are instances where the memoryless optimum is close to being a factor of $2$
larger than the respective stochastic optimum.

\begin{lemma}  \label{SMgaps:example}
 For any $\epsilon>0$, there is an instance on which
\begin{align*}
 \text{opt}_M\text{-\MAXe} = \text{opt}_M\text{-\SUMe} \geq
 (2-\epsilon)  \text{opt-\MAXe} = \text{opt-\SUMe}
\end{align*}
\end{lemma}
\begin{proof}
The instance has $n$ elements, corresponding $n$ singleton tests, and
uniform priorities $p_e$. The optimal
  memoryless schedule, the solution of both  \eqref{basicLP} and
  \eqref{basicconvex}, has $q_e=1/n$ and $\Mt[e]=\Et[e]=n$ for each element.
The optimal deterministic   schedule repeats a permutation on the $n$
elements and has $\Mt[e]=\Me[t]=n$ and $\Et[e]=\Ee[t]=(n+1)/2$ for all
$e,t$.
The optimal stochastic
selects a permutation uniformly at random every $n$ steps and follows
it.  It has $\Mt[e]=\Et[e]=(n+1)/2$ for all elements.
\end{proof}

\section{Deterministic scheduling} \label{deterministic:sec}

 The distinction between objectives within each of the \MAXe\ and
 \SUMe\ groups does matter  with  deterministic scheduling.
  For an instance and objective, we attempt to understand the relation between
the deterministic and stochastic optima.
For deterministic \MAXe\ objectives, the comparison is to  opt-\MAXe\ and for \SUMe\
objectives, it is to opt-\SUMe.

  We show that on all instances, the deterministic
\EEP\ is equal to  opt-\SUMe.
Deterministic \EMP\ and \WMP, however, are always strictly
larger (proof is provided in \notinproc{\ignore{Appendix}
  \ref{svsdlemmaproof:sec}.}\onlyinproc{\cite{blackholes:arxiv2013}.}
\begin{lemma} \label{SvsDlemma}
\begin{align}
\text{opt$_D$-\EEP}  &= \text{opt-\SUMe}  \label{EEPSD:claim} \\
\text{opt$_D$-\EMP} &\geq 2\text{opt-\SUMe} -1  \label{EMP:claim} \\
\text{opt$_D$-\WMP} &\geq 2\text{opt-\MAXe} -1 \label{WMP:claim}
\end{align}
\end{lemma}
We can show that finding the optimal schedules for all these objectives is
NP hard.
\begin{lemma} \label{NPhard-det}
Computing any one of the following optima is NP hard:
opt$_D$-\EEP, opt$_D$-\MEP, opt$_D$-\EMP, 
 opt$_D$-\WMP,  opt$_D$-\WEP,  and opt$_D$-\MWP.
\end{lemma}
This proof, similarly to the proof of  Lemma \ref{NPhard}, is also based on a reduction to the to exact cover by sets of size 3 (X3C) and details 
are omitted.  Additional relations, upper bounding the deterministic
objective by the stochastic objective
 follow from relations with memoryless optima which are presented next.

  For a deterministic schedule and an objective, the
  {\em approximation ratio}
is the ratio of the objective on the
schedule to that of the (deterministic) optimum of the same objective.
We are ultimately interested in efficient constructions of deterministic
schedules with good approximation ratio and in quantifying the cost of
determinism, that is, asking how much worse a deterministic objective can be over
the respective stochastic objective.

We define the {\em D2D}, {\em D2M}, and {\em D2S} of
a deterministic schedule as the ratio of the objective on the
schedule to that of the deterministic, memoryless, or stochastic  optimum of the same objective.
Since both deterministic and stochastic
optima are NP hard to compute, so is the D2D (the approximation ratio) and the
D2S.
The D2M of any given schedule, however, can be computed efficiently by computing
the memoryless optimum.  The D2M can then be used to bound the D2D and
D2S, giving an upper bound on how far our schedule is from the optimal
deterministic or stochastic schedule.  In particular,
the relation $\text{D2S} \leq \text{D2M} \leq
 2\ \text{D2S}$ follows from Theorem~\ref{memvssto:thm}.
We study the relation between the memoryless and deterministic optima below.

\subsection{Memoryless versus deterministic}  \label{memvdet:sec}
Since a deterministic schedule is a special case of a
stochastic schedule, from Theorem~\ref{memvssto:thm}, the
memoryless optimum is at most twice the deterministic optimum.
The  proof of Lemma \ref{SMgaps:example} shows:
\begin{lemma}\label{DSMgaps:example}
For any $\epsilon$, there is an instance on which
\begin{align*}
A = & \text{opt}_M\text{-\MAXe} = \text{opt}_M\text{-\SUMe} =
\text{opt$_D$-}\MeMt = 
\text{opt$_D$-}\EeMt  = \text{opt$_D$-}\EtMe
        \\
B=  & \text{opt$_D$-}\MeEt = \text{opt$_D$-}\EeEt =
 \text{opt$_D$-}\MtEe =  
\text{opt-\MAXe} = \text{opt-\SUMe}\\
& A \geq (2-\epsilon) B
\end{align*}
\end{lemma}
That is,  for the weaker
\SUMe\ and \MAXe\ deterministic
objectives, a gap of 2 is indeed realizable, meaning that it is possible
for the deterministic optimum to be smaller than
the respective memoryless optimum.
For the strongest objectives, \EMP\ for \SUMe\ and $\MeMt$ for \MAXe,
 we show that the
deterministic optimum is at least the memoryless optimum:
\EC{ We have a ``hole'' for  $ \EtMe$  !  can deterministic be better
  than memoryless ? if so, what is the gap ? is there an
  example where there is a gap of (near) 2 ?  (the example below does
  not show a gap since all times $t$ are symmetric.}
\begin{lemma} \label{memdetupper}
\begin{align*}
\text{opt$_M$-\SUMe}  \leq  \text{opt$_D$-\EMP}\\
\text{opt$_M$-\MAXe}  \leq  \text{opt$_D$-}\MeMt
\end{align*}
\end{lemma}
\begin{proof}
Similar to the proof of Theorem~\ref{memvssto:thm}:
Consider a deterministic schedule and let
$q_i$ be (the limit of) the relative frequency of test $i$.
We have
$\Mt[e] \geq  \frac{p_e}{\sum_{i|e\in \test_i} q_i}$.
\end{proof}

 We next consider the other direction, upper bounding the deterministic optimum by
the memoryless optimum.
  For the objectives $\EeEt$ and $\MeEt$, which are respectively the weakest
  \SUMe\ and \MAXe\ objectives, we
show that the deterministic optimum is at most the
memoryless optimum.  Moreover, we can efficiently construct
deterministic schedules with  D2M arbitrarily close to $1$ (and thus
approximation ratio of at most $2$).

\begin{lemma}  \label{easymeetNeeet}
\begin{align*}
\text{opt$_D$-}\MeEt & \leq \text{opt$_M$-\MAXe} \\
\text{opt$_D$-}\EeEt & \leq \text{opt$_M$-\SUMe}
\end{align*}
and for any $\epsilon>0$ we can efficiently construct deterministic
schedules with $\MeEt$  or $\EeEt$ D2M $\leq (1+\epsilon)$.
\end{lemma}
\begin{proof}
 For any $\epsilon$, for a long enough run of
the memoryless schedule $\boldsymbol{q}$, there is a positive
probability that for all elements, the average over time of $T(e,t)$
(in the part of the sequence where it is finite) is at most
$(1+\epsilon)\Et[e | \boldsymbol{q}]$. We obtain the deterministic
schedule by cycling through such a run. If the run is sufficiently
long  then the suffix in which  $T(e,t)$ is infinite is a small
fraction of the run and the resulting
 schedule $\sigma$ has $\MeEt[\sigma] \leq (1+\epsilon)
\text{opt$_M$-\MAXe}$ and $\EeEt[\sigma] \leq (1+\epsilon)
\text{opt$_M$-\SUMe}$.
\end{proof}

\EC{  What can we say about the approximation ratio of $\MtEe$ ?   we
  do not have a bad example where it is more than a constant, but do
  not have a proof that it is a constant.}

We are now ready to relate the D2D and D2M.  We obtain
$\text{D2D} \leq 2\ \text{D2M}$, and for \EMP\ and $\MeMt$ (see Lemma ~\ref{memdetupper}),
we have $\text{D2D} \leq \text{D2M}$.
  Accordingly,  D2M $\geq 1/2$,  and for
\EMP\ and $\MeMt$ we have D2M $\geq 1$.  The optimum D2M  is
the minimum possible over all schedules.  We refer to the supremum of
optimum D2M over instances as the {\em D2M gap}
of the scheduling problem.


In contrast, for the strongest objectives,  $\MeMt$, $\MeEt$, and
$\EeMt$, we construct a family of instances with asymptotically large optimal
D2M, obtaining a lower bound on the D2M gap.   We also show that
$\text{opt$_D$-}\MeMt$ and
$\text{opt$_D$-}\EtMe$ are hard to approximate better than $\ln(\nElements)$:
(Proof details are provided in \onlyinproc{\cite{blackholes:arxiv2013}}\notinproc{\ignore{Appendix~}\ref{lowerboundlink:sec}})
\EC{Can we show asymptotic hardness of approximation  for $\EeMt$?
  can we show either hardness of approximation or asymptotic
  D2M ratio for $\MtEe$ ?}

\begin{lemma} \label{lowerboundlink}
There is a family of instances with $\nTests$ tests and $\nElements$
elements such that each element participates in $\nTofe$ tests with
the following lower bounds on D2M:
The \MWP-D2M  (and thus \WMP-D2M)
$\Omega(\ln \nElements)$ and $\Omega(\nTests)$. 
The \EMP\ D2M  is $\Omega(\log \nTofe)$.
Moreover, these instances can be realized on a network, where elements
are links and tests are paths.
\end{lemma}

  \begin{lemma}  \label{setcoverlb}
The problems
$\text{opt$_D$-}\MeMt$ and $\text{opt$_D$-}\EtMe$ are hard to
approximate to anything better than
$\ln(\nElements)$.
  \end{lemma}
\begin{proof}
When $\boldsymbol{p}$ is uniform,
$\text{opt$_D$-}\MeMt$ is equivalent to set cover -- an approximation
ratio for $\text{opt$_D$-}\MeMt$ implies the same approximation ratio
for set cover~\cite{NguyenTTD:infocom09}, which is hard to
approximate~\cite{feige98}.

  This also extends to $\text{opt$_D$-}\EtMe$, again using uniform
  $\boldsymbol{p}$.  A minimum set cover of size $k$
  implies a schedule (cycling through the cover) with $\EtMe$ of $k$.
Also, a schedule with $\EtMe$ at most $k$ means that $\Me[t]\leq k$
for at least one $t$, means there is a cover of size $k$.
\end{proof}

\subsection*{Summary of relations}
A summary of these relations, which also includes results from
our R-Tree schedulers (Section \ref{treeschedules:sec})  is provided
in Table~\ref{D2M:table}.
The lower bounds on the D2M gap are established in Lemma~\ref{lowerboundlink} through
example instance on which the optimum D2M is large.
The lower bound on approximability is established in
Lemma~\ref{setcoverlb}.  Both lower and upper bounds
 for $\EeEt$ and $\MeEt$ are established in Lemma~\ref{easymeetNeeet},
\EMP\ D2M upper bound in Theorem~\ref{thm:EMP},  and $\MeMt$ D2M in
Theorem~\ref{lemma:treememt}. 

\begin{table} [h]
{\small
\begin{tabular}{l|l|l| l}
objective  &   scheduling D2M   & D2M gap & approximability \\
\hline
$\EtEe$ & 1  &  1 &  \\
$\MeEt$ & 1  & 1 &   \\
$\EeMt$ & $O(\ln \nTests)$ &  $\Omega(\ln \nTests)$ & \\
$\MeMt$, $\EtMe$ & $O(\log \nElements + \log \nTofe)$ &
$\Omega(\log \nElements)$, $\Omega(\nTests)$ & $\Omega(\log \nElements)$ \\
\hline
\end{tabular}
}
\caption{D2M upper bounds of our schedulers
and lower bounds on the
D2M gap and on efficient approximability.\label{D2M:table}}
\end{table}

\ignore{
We show that

\noindent
(i)  There are instances
(with $\nTests$ tests,  $\nElements = {\nTests \choose \nTofe}$ elements, and exactly $\nTofe$ tests
including each element)  where the \MWP\ D2M-ratio (and
therefore also the \WMP\ D2M-ratio)
on any schedule is $\Omega(\nTests)$ and the
  \EMP\ ratio on any schedule is $\Omega(\nTofe)$.

\noindent
(ii)
For any instance we can efficiently construct a deterministic schedule with
\EMP\ ratio  $O(\max_e \log k_e)$, where
$k_e$ is the maximum number of tests covering at least a constant
fraction of the frequency of element $e$ in the optimal memoryless
schedule.
(This is bounded by the maximum number of tests  containing a
particular element)
Moreover, the same schedule has constant \EEP\ ratio.
}


\section{R-Tree schedules} \label{treeschedules:sec}

   We present an efficient construction of deterministic schedules
   from a distribution
 $\boldsymbol{q}$ and relate detection times of the deterministic schedule
to (expected) detection times of the memoryless schedule defined by
$\boldsymbol{q}$.

We can tune the schedule to  either \MAXe\
or \SUMe\
objectives, by selecting accordingly the input frequencies
$\boldsymbol{q}$ as a solution of \eqref{basicLP} or
\eqref{basicconvex}.
We then derive analytic bounds on the D2M of the schedules we obtain.

 The building block of {\em random tree}  (R-Tree) schedules
is {\em tree schedules}, which are deterministic
schedules specified by a mapping of tests to nodes
of a binary tree.   A tree schedule is specified with respect to
probing frequencies $\boldsymbol{q}$ and has the property that for any test,
the maximum probing interval in the deterministic schedule
is guaranteed to be close to $1/q_i$.
However, if we do not place the tests in the tree carefully then
for an element covered by multiple tests the probing
interval can be
close to that of its most frequent test, but yet far from  the desired
(inverse of)
$Q_e=\sum_{i|e\in \test_i} q_i$.  Therefore,  even when computed with
respect to
$\boldsymbol{q}$ which solves \eqref{basicconvex},
the tree schedule can have  \EMP\ and \EEP\ D2M ratios $\Omega(\nTofe)$.

We define a distribution over tree schedules obtained by
randomizing the mapping of
tests to nodes.  We then bound the expectation of the \EMP\ and
\EEP\ (when applied to $\boldsymbol{q}$ which solves
\eqref{basicconvex})  and
$\MeMt$ (when applied to $\boldsymbol{q}$ which solves
\eqref{basicLP}) over the resulting deterministic schedules.
Given a bound on the expectation of an objective, there is a constant
probability that a tree schedule randomly drawn from the distribution
will satisfy the same bound (up to a small constant factor).
An R-Tree  schedule  is obtained by constructing multiple tree schedules drawn from the
distribution, computing the objectives on these schedules, and
finally, returning the best performing tree schedule.  Note that even though
the construction is randomized, the end result, the R-Tree schedule,
is deterministic, since it is simply a tree schedule.

Specifically,  lets take  \SUMe as an example,  we
apply the R-Tree schedule construction several times with  $\boldsymbol{q}$'s
solving \eqref{basicconvex}.
The tree with the best  \EMP\ has   $O(\log(\nTofe))$   \EMP\ D2M
 and the tree with the best \EEP\ has
a constant \EEP\ D2M. Furthermore we
can also find  a tree  which satisfies both guarantees.

\begin{theorem}
\label{thm:EMP}
A deterministic schedule with
\EMP\ D2M ratio of $O(\log \nTofe)$ and
a constant \EEP\ D2M ratio can be constructed efficiently.
\end{theorem}
The theorem is tight since from Lemma~\ref{lowerboundlink},
 the \EMP\ D2M gap on some instances  is $\Omega(\log \nTofe)$, and
 therefore, we can not hope for a better dependence on $\ell$.\footnote{As a side note, recall that according to
 \eqref{EEPSD:claim} there exist schedules with \EEP\ D2M close to
 $1$, so with respect to \EEP\ this only shows that we can
simultaneously obtain a \EMP\ D2M that is logarithmic in $\nTofe$
and at the same time a constant \EEP\ D2M.}


 For \MAXe,
we show that when  we apply the R-Tree schedule construction to $\boldsymbol{q}$ which is the
    optimum of \eqref{basicLP},  we obtain a deterministic schedule with
$O(\log \nTofe + \log \nElements)$ $\MeMt$ D2M.
\begin{theorem}  \label{lemma:treememt}
A deterministic schedule with
$\MeMt$  D2M ratio of  $O(\log \nTofe + \log \nElements)$ can be constructed efficiently.
\EC{Can we add constant $\MeEt$ D2M ratio here ?}
\end{theorem}

From Theorems~\ref{thm:EMP} and~\ref{lemma:treememt}, we obtain the
following upper bounds on the D2M gap and efficiently construct
deterministic schedules satisfying these bounds
(summarized in Table~\ref{D2M:table}).
\begin{align*}
\text{opt$_D$-\EMP}  &= O(\log
\nTofe)  \text{opt-\SUMe} \\
\text{opt$_D$-}\MeMt &= O(\log \nTofe + \log\nElements)
\text{opt-}\MeMt
\end{align*}

We provide construction details of our R-Tree schedulers.  The
analysis, which includes the proofs of Theorems~\ref{thm:EMP} and~\ref{lemma:treememt}, is deferred to \ignore{Appendix }\ref{rtreedetails:sec}.
\subsection{Tree schedules} \label{treeschedules}

 A tree schedule is a deterministic schedule guided with frequencies
 $\boldsymbol{q}$ where probes to test
$i$ are spaced $[1/q_i,2/q_i)$ probes apart.
 When $q_i$ has the  form $q_i=2^{-j}$,  test $i$ is
performed regularly with period $2^j$.

 Assume for now that $q_i=2^{-L_i}$ for positive integer $L_i$ for all $i$.
 We map each $i$ to nodes  of a binary tree where $i$ is mapped to a
 node at level $L_i$ and no test can be a child of another.
 This can be achieved by greedily mapping
tests by decreasing level -- we greedily map tests with level $L_i=1$,
then tests with $L_i=2$ and so on.  Once a test is mapped to a node,
its subtree is truncated and it becomes a leaf.

 From this mapping, we can generate a deterministic schedule
as follows:  The sequence
is built on alternations between left and right child at each node.  Each node
``remembers'' the last direction to a child.
To select a test, we do as follows.  First visit the root and select the child that was not visited previous time.
If a leaf, we are done, otherwise, we recursively select the child that was not previously visited and continue.  This
until we get to a leaf.  We then output test $i$.  This process
changed ``last visit'' states on all nodes in the path from the root to
the leaf.  It is easy to see that if a leaf at level $L$ is visited
once every $2^L$ probes.  An example of a set of frequencies, a corresponding mapping,
and the resulting schedule is provided in Figure~\ref{treesched:fig}.

  If probabilities are of general form, we can map each test according to the highest
  order significant bit (and arbitrarily fill up the tree).  When doing this we
  get  per-test ratio between the actual and desired probing
  frequencies of at most $2$.
 Alternatively, we can look at the  bit
  representation of $q_i$--  separately map all ``$1$'' positions in
  the first few significant bits to tree nodes. In this
  case the average probing frequency of each test is very close
  to  $q_i$ but the maximum time between probes depends on the relation
between the tree nodes to which the bits of test $i$ are mapped to.
The only guarantee we have on the maximum
is according to the most significant bit
$2^ {-\lceil\log_2 (1/q_i) \rceil}$.
Under ``random''  mappings the
expectation of the maximum gets closer to the average.

\begin{figure}[t]
\centering
\begin{minipage}{2in}
\ifpdf
\includegraphics[width=0.85\textwidth]{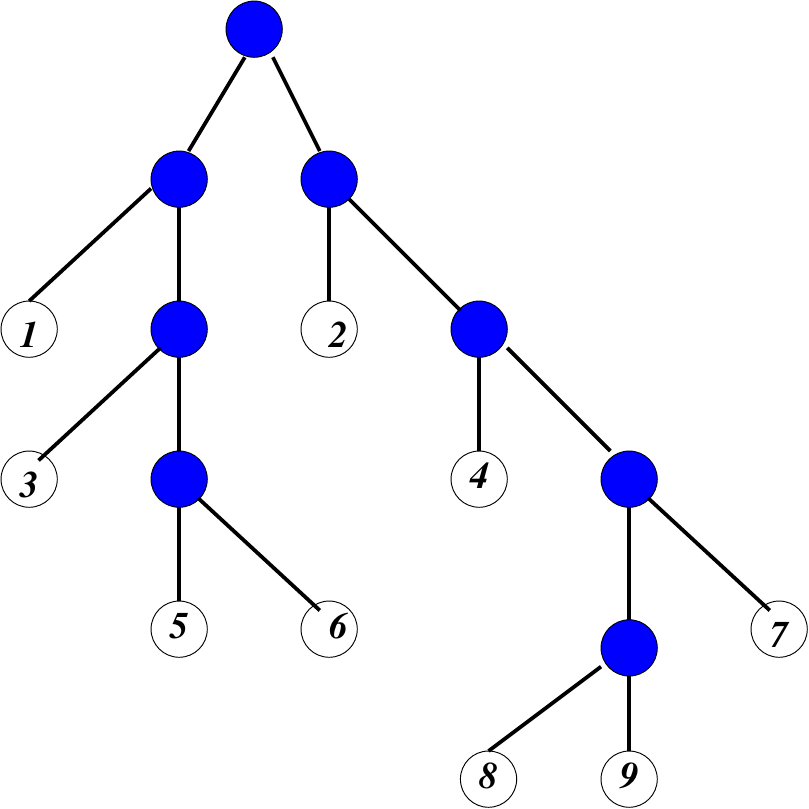}
\else
\epsfig{figure=tree.eps,width=0.25\textwidth}
\fi

\end{minipage}
\begin{minipage}{4in}
{\small

$q_1=q_2=1/4$, $q_3=q_4=1/8$, $q_5=q_6=q_7=1/16$, $q_8=q_9=1/32$

\smallskip
\begin{tabular}{c|rrrrrrrrrrrrrrrr}
\hline
L & \multicolumn{16}{l}{schedule} \\
\hline
2 & 1 & 2 & x  & x &   &   &   &   &   &   &   &   &   &   &   &  \\
\hline
3 & 1 & 2 & 3  & 4 & 1 & 2 & x & x &   &   &   &   &   &   &   &  \\
\hline
4 & 1 & 2 & 3  & 4 & 1 & 2 & 5 & x & 1 & 2 & 3  & 4 & 1 & 2 & 6 & 7 \\
\hline
5 & 1 & 2 & 3  & 4 & 1 & 2 & 5 & 8 & 1 & 2 & 3  & 4 & 1 & 2 & 6 & 7 \\
- & 1 & 2 & 3  & 4 & 1 & 2 & 5 & 9 & 1 & 2 & 3  & 4 & 1 & 2 & 6 & 7 \\
\hline
\end{tabular}
}
\end{minipage}
\caption{Mapping tests to nodes of a binary tree to produce a
  deterministic schedule.   The table shows the level-L schedule for
  $L=2,3,4,5$.   The full deterministic schedule cycles through the
  level-5 schedule. \label{treesched:fig}}
\end{figure}

\subsection{Random tree schedules}

Consider an instance and a memoryless schedule with frequencies
$\boldsymbol{q}$. We assume that $q_i$ have the form $2^{-L_i}$ for
positive integers $L_i$ (this is without loss of generality as we
can only look at the highest order bit and loose a factor of at most
2). We construct a tree schedule for $\boldsymbol{q}$ by mapping the
tests  to nodes randomly as follows.  We process tests by increasing
level. In each  step (level), all tests of the current level are
randomly mapped to the available tree nodes at that level.  After a
test is mapped to a node, its subtree is truncated.

  For each level $N$ (which can be at most the maximum  $L_i$), we can
  consider the {\em level-$N$ schedule}, which is a cyclic schedule of
  length $2^N$.   The schedule specifies the probes for all tests with
  level $L_i \leq N$, and leaves some spots ``unspecified''.

 We now specify the level-$N$ schedule of the tree. Consider a
completion of the tree to  a full binary one with $2^N$ leaves
(truncate everything below level $N$).
 Associate with each leaf $a$
a binary number $\overline{a}$ which contains a $0$ at digit $i$
(from right to left, i.e.\ the least significant  digit correspond
to the child of the root and the most significant digit corresponds
to the leaf itself)
 if the $i$th child on its path from the root is a left
child. We refer to $\overline{a}$  as the {\em position} of leaf
$a$.

We construct the sequence  by associating  test $i$ with all leaf
descendants  of the node containing it, and with  all the positions
of the sequence corresponding to these leaves. Putting it in another
words the level-$N$ schedule of the tree cycles through the leaves
$a$ at level-$N$ (of the completion  of the tree) according
 to the order defined by $\overline{a}$  and probes the test associated with each leaf.
A test with $q_i=2^{-L_i}$ is probed in regular intervals of
$2^{L_i}$.  The first probe is distributed uniformly at random from
$[0,2^{L_i}-1]$.

Level-$N$ schedules constructed from the same mapping  for different depths
  $N$ are consistent in the following sense:
The level
$N'> N$ schedule is $2^{N'-N}$ repetitions of the level-$N$ schedule
in terms of the tests specified by a level-$N$ schedule (those with
level $L_i\leq N$)
 and also specifies tests with $N< L_i\leq N'$.

\section{The Kuhn-Tucker scheduler} \label{greedy:sec}

The Kuhn-Tucker conditions on the optimal solution of our convex
program \eqref{basicconvex} imply that the values
$$r_i=\frac{\partial \sum_e \frac{p_e}{\sum_{i|e\in \test_i}
    q_i}}{\partial q_i} = -\sum_{e | e\in \test_i}
\frac{p_e}{(\sum_{j|e\in \test_j} q_j)^2}\ .$$
are balanced for different tests.
Based on that,
we suggest a deterministic greedy heuristic for \SUMe,
illustrated in Algorithm~\ref{kuhntucker:alg}.
For each element $e$, we track $x[e] \geq 1$ which is the elapsed number of
probes since $e$ was last probed.
We then choose the test $i$ with maximum
$\sum_{ e\in \test_i}  p_e x[e]^2$.

We conjecture that the KT schedule has $\EtEe$ which is at most twice the optimal.
 Viewing the quantity $\sum_e p_e  x[e]^2$ as ``potential''  the average
 reduction in potential is the $\EtEe$ of the sequence.
We do not
provide bounds on the approximation ratio, but test this heuristic
in our experiments.

\begin{algorithm} [h]
\begin{minipage}{3.8in}
\begin{algorithmic}
\Function{{\sc best-test}}{}
  \State $v\gets 0$
    \For {$\test\in \Tests$}
     \State $y\gets 0$
     \For {$e\in \test$}
        \State $y\gets y+ p_e x[e]^2$
     \EndFor
     \If {$(y>v)$}
       \State $b\gets \test$; $v\gets y$
     \EndIf
    \EndFor
  \Return $b$ \Comment{test with maximum $\sum_{ e\in \test_i}  p_e x[e]^2$ }
\EndFunction
\end{algorithmic}
\end{minipage}
\begin{minipage}{2.5in}
\begin{algorithmic}
\Function{{\sc KT-schedule}}{$\Elements,\boldmath{p},\Tests$}
    \For {$e \in \Elements$}
      \State $x[e]\gets 1$
    \EndFor
\While {True}
\State  $\test \gets ${\sc best-test}$()$
\State {\bf output}  $\test$
\State For $e$  let $x[e] \gets x[e]+1$
\For {$e\in \test$}
\State $x[e]\gets 1$
\EndFor
\EndWhile
\EndFunction
\end{algorithmic}
\end{minipage}
\caption{Kuhn-Tucker (KT) schedule \label{kuhntucker:alg}}
\end{algorithm}

 The KT scheduler can be deployed when priorities are modified on the go.
 This is in contrast to other schedulers which pre-compute the
 schedule .

\EC{
  Analyses:  The sum of the reductions in potential is the \EEP\ of
  the schedule we are generating.  Can we show it is within a constant
of the minimal \EEP?  can we show anything on the \EMP ?
This is simple and may perform well in practice, as it does not loose
the constants of the tree schedule.
}

\section{Experimental Evaluation}  \label{exper:sec}

\begin{table*}
\mbox{

\begin{tabular}{l|l|l|l|l}
\multicolumn{5}{c}{{\bf  \SUMe\ in memoryless schedulers:}}\\
algorithm  &   GN-U & GN-P & GN-Z  &  Clos \\
\hline
\hline
Convex & 95.66 &  59.72 & 25.92 & 32.02  \\
LP & 105.29 & 68.77 & 118.38 & 32.02 \\
Uniform & 229.16 & 72.27 & 260.46 &  33.00 \\
\hline
SAMP SC & 111.56 &  82.70 & 86.17& 32.00  \\
SAMP KT & 108.54 & 61.45 & 86.17 & 32.00 \\
\hline
\hline
\end{tabular}

$\quad\quad$

\begin{tabular}{l|l|l|l|l}
\multicolumn{5}{c}{{\bf $\EtEe$ in deterministic schedulers:}}\\
algorithm  &   GN-U & GN-P & GN-Z  &  Clos \\
\hline
\hline
SC     & 60.27  &  49.42 & 51.52 & 16.50 \\
KT     & 58.04  &  33.93 & 14.63 & 16.50 \\
RT CON & 66.43  &  49.21 & 17.92 & 30.76 \\
RT LP  & 85.47  &  63.81 & 88.07 & 31.00 \\
\hline
RT-S CON & 57.87  &  46.91 & 18.69 &    \\
RT-S LP  & 59.70  &  50.24 & 87.47 &    \\
\hline
\hline
\end{tabular}
}

\vspace{0.3in}

\mbox{
\begin{tabular}{l|l|l|l|l}
\multicolumn{5}{c}{{\bf $\MtEe$ in deterministic schedulers:}}\\
algorithm  &   GN-U & GN-P & GN-Z  &  Clos \\
\hline
\hline
SC     & 70.43  &  62.08 & 93.80 & 16.50  \\
KT     & 70.04  &  62.08 & 20.36 & 16.50  \\
RT CON & 72.23  &  56.53 & 24.65 & 36.05 \\
RT LP  & 95.81  &  73.34 & 113.47& 36.20 \\
\hline
RT-S CON & 60.08  &  50.02 & 23.38 &    \\
RT-S LP  & 63.09  &  53.82 & 96.67 &    \\
\hline
\hline
\end{tabular}

 $\quad\quad$

\begin{tabular}{l|l|l|l|l}
\multicolumn{5}{c}{{\bf $\EeMt$ in deterministic schedulers:}}\\
algorithm  &   GN-U & GN-P & GN-Z  &  Clos \\
\hline
\hline
SC     & 124.93  &  109.46 &  114.29 & 32.00  \\
KT     & 130.11  &  93.02  &  35.34  & 32.00   \\
RT CON & 180.00  &  179.91 &  53.40   & 144.14 \\
RT LP  & 319.12  &  261.65 &  269.01 & 146.70 \\
\hline
RT-S CON & 121.24  &  103.91 & 42.61  &    \\
RT-S LP  & 123.35  &  107.42 & 183.89 &    \\
\hline
\hline
\end{tabular}
}
\caption{\SUMe\ objectives.  Table shows expected time with memoryless
schedules (same for all \SUMe\ objectives) and $\EeEt \leq \MtEe \leq
\EeMt$ on different deterministic schedulers. \label{sume:table}}
\end{table*}


We evaluated the performance of our schedulers for testing for silent
link failures in two networks. The first is a
 backbone network (denoted GN in the sequel) of a large enterprise.
We tested 500 of the network links with 3000 MPLS paths going through them.

The second network we considered is a (very regular) folded Clos
network (denoted Clos) of 3 levels and 2048 links.  On this network we
considered all paths between endpoints.  The Clos network is a
typical interconnection network in data centers.

\ignore{
is detecting  silent hardware failures in
routers.  This is  modeled by representing each router with
$k$ interfaces as a full $k$ bipartite graph and considering the network over
these interfaces where links connect
appropriate interfaces of different routers.  The paths we consider
are all ``routable'' packet trajectories.
For this scenario we use
a (very regular) folded Clos network~\footnote{See http://en.wikipedia.org/wiki/Clos_network.}
(referred to as Clos) of 3 levels and 2048
links, together with  all possible paths between endpoints.
}





 For the Clos network, we only considered uniform weights (priorities), 
meaning that all links are equally important.
For the GN network, we considered uniform weights (denoted GN-U), 
weights that are 
 proportional to the number of MPLS paths traversing the link (GN-P,
 where P designates popularity), and 
Zipf distributed weights with parameter 1.5 (GN-Z).


On these four networks (links and paths with associated weights),
Clos, GN-U, GN-P, and GN-Z, 
we simulated our schedulers and evaluated their performance with respect to
the different objectives.

\medskip
\noindent
{\bf Memoryless schedulers:}
We  solved the 
convex program \eqref{basicconvex} for \SUMe\ objectives  and the LP
\eqref{basicLP} for \MAXe\ objectives to obtain optimal memoryless
probing frequencies $\boldsymbol{q}$.
These optimization problems were solved using

Matlab (for the LP) and CVX (for the convex program, see {\tt
  http://cvxr.com/cvx/}).

  We compared these optimal memoryless schedules to other memoryless
schedules obtained using three naive selections of probing frequencies:  the first is
uniform probing of all paths (Uniform), the second is uniform probing of a 
smaller set of paths that cover all the links (SAMP SC), and the third is probing
according to frequencies generated by the Kuhn-Tucker schedule (SAMP
KT).

 The performance of these schedules, in terms of the expected
 detection times $\text{T}(e,t)$ is  shown in Table~\ref{sume:table}
 (\SUMe\ objective) and Table~\ref{maxe:table} (\MAXe\ objective).
The schedulers optimized  for one of the 
objectives,  \SUMe\ or  \MAXe,   clearly dominate all others with
respect to the objective it optimizes.
We can see that 
while on some instances the alternative schedulers perform close to
optimal, 
performance gaps can sometimes be substantial.  In particular,  a
schedule optimized for one objective can perform poorly 
with respect to the other objective.
We note, however, that our unified treatment 
facilitates designing schedules which trade off performance with
respect to two objectives.



We illustrate the qualitative difference  between the \SUMe\ and \MAXe\
objectives through Figure~\ref{sumvmax:fig} (A).  The figure shows  a reverse CDF
of $\text{T}(e,t)$, the expected time to detect a failure of a link of the
backbone network with uniform weights (GN-U).
(Recall that $\text{T}(e,t)$ is fixed for all $t$ for memoryless
schedules.)
Given a reverse CDF of a schedule, 
the maximum point on the curve is the \MAXe\ of the schedule
whereas the average value (area under the curve)  is the \SUMe\ of the schedule.
 We can see that the schedule computed by the LP \eqref{basicLP},
 which optimizes \MAXe\ has a smaller maximum whereas the schedule
computed by the convex program \eqref{basicconvex} has a smaller area.

\medskip
\noindent
{\bf Deterministic schedulers:}
We now evaluate our deterministic schedulers.   Here, $\text{T}(e,t)$, the elapsed time from time $t$  till the next path
 containing $e$ is scheduled,  is deterministic.
 We used two different implementation of the R-Tree
algorithm (Section~\ref{treeschedules:sec}).  In the first, the algorithm
was seeded with the frequencies computed by the LP (RT LP) or by the convex program
(RT CON) when applied to the full set of paths.  We discuss the second implementation
in the sequel.
 We also implemented
 the Kuhn-Tucker (KT)
 scheduler (Section~\ref{greedy:sec}), and the classic greedy Set Cover
 algorithm (SC) which was previously used for the $\MeMt$ metric \cite{ZhengCao:IEEEtComp2012,NguyenTTD:infocom09,ZKVM:conext2012}
 (minimum set cover is the optimal deterministic scheduler for $\MeMt$ when priorities are uniform).
 This scheduler cycles through a
 sequence consisting of this set cover.

Table \ref{sume:table} shows the values of all \SUMe\ objectives for
the different memoryless and deterministic schedulers and Table
\ref{maxe:table} shows the same for the \MAXe\ objectives. It is easy
to verify the relations between the three different \SUMe\ objectives
and three different \MAXe\ objectives
(see Lemma \ref{basicrel}).   The gaps between the objectives show
again that an informed selection of the objective is important.
We can also see that with uniform priorities (GN-U and Clos) 
the SC scheduler performs well.  Indeed, 
in this case
minimum set cover
produces the optimal deterministic schedule for $\MeMt$ and $\EtMe$. 
  When priorities are
highly skewed, however, as is the case for GN-Z, its performance deteriorates.

The KT scheduler performed well on the \SUMe\ objectives, which it is
designed for.  Because of its adaptive design,  which does not involve
precomputation of a fixed schedule, the KT scheduler is highly suitable for applications where priorities
are changing on the go.  One such scenario is when priorities of
different elements correspond to the current
traffic levels traversing the element.  The KT scheduler gracefully
adapts to changing traffic levels.

Our R-Tree schedulers (RT CON and RT LP) did not perform well on some of the instances, 
and in some cases, performed worse than 
SC and KT.   The reason, as the analysis shows (see Section \ref{treeschedules:sec}), is the logarithmic dependence on
$\ell$, which in our case, is the maximum number of paths used to cover an element
in the solution of the LP and convex programs.   The collection of
paths computed by the LP and Convex solvers 
turned out to have  high redundancy, where subpaths have many
alternatives and the fractional solvers tend to equally use all
applicable paths.  We can see evidence for this fragmentation in Figure  \ref{sumvmax:fig}.


To address this issue,  we seeded the R-Tree algorithm with respective solutions
of the LP and Convex programs applied to a modified instance with 
a pre-selected small subset 
of the original paths.  The subset was picked so that it contains a cover of the links and
also tested to ensure that  the objective of the optimization problem does not
significantly increase when implementing this restriction.   On those instances, 
tests which constitute a set cover of the links and
produced by the greedy approximation algorithm, performed well.
We denote the respective schedulers obtained this way using the LP and convex solutions,
 by RT-S LP and RT-S CON.



The results of this experiment  are included in
Tables \ref{sume:table} and \ref{maxe:table}. We
can observe that this heuristic substantially improves the performance of the R-Tree
algorithm  for all objectives.  Moreover, RT-S was never worse than
SC, and when SC was not optimal, substantially improved over SC.
We  leave the question
of how to choose the subset to best balance the loss in the objective of the
memoryless schedule with the gain in better derandomization for
further research.



\medskip
\noindent
{\bf Memoryless vs. Deterministic:}
Memoryless schedulers are stateless and highly suitable for
distributed deployment whereas deployment of
deterministic schedulers requires some coordination between probes 
initiated from different start points.
 However, due to their stochastic nature, with memoryless scheduling
we can only obtain guarantees on the expectation whereas 
with deterministic schedulers we can obtain
worst case guarantees on the time (or weighted cost) until a failure
is detected.
We demonstrate this issue by illustrating, in
 Figure~\ref{detVsRTree:fig} (B)  the distribution over the
links of the backbone graph of the maximum detection time in the
deterministic R-Tree scheduler, $\Mt[e]$, and the 99th percentile
line  for the
memoryless schedulers (elapsed time to detection in
99\% of the time).
Figure \ref{detVsRTree:fig} (C) shows the same data for the
schedulers RT-S LP and RT-S CON which were derived after restricting
the set of paths over which optimization was performed.
 One can see that when there are strict requirements on  worst-case detection
 times,  deterministic schedules dominate.

Moreover, even when comparing expected (memoryless) versus worst-case
(deterministic) detection times, we can see that our best deterministic schedulers often have
$\EeEt$, $\MtEe$, and $\MeEt$ 
detection times that are $20\%-50\%$ smaller
than the respective memoryless optimum.  
Our analysis shows 
(Section \ref{memvdet:sec}) that on these objectives it is possible
for the optimal deterministic
detection times to be up to a factor of 2 smaller
than the respective memoryless optimum.   On the remaining
objectives, the deterministic optimum can not be better than the
memoryless one and can be much worse (asymptotically so).
Recall that 
while the memoryless optimum can be precisely computed, the
deterministic optimum  is NP hard to compute (Lemma \ref{NPhard}).
Therefore, these relations tell us that in many cases our best deterministic
schedules obtained nearly optimal schedules.


%

\begin{table*}

\mbox{
\begin{tabular}{l|l|l|l|l}
\multicolumn{5}{c}{{\bf  \MAXe\ in memoryless schedulers:}}\\
algorithm  &   GN-U & GN-P & GN-Z  &  Clos \\
\hline
\hline
Convex & 221.53 &  21.81 & 6.85 & 32.02  \\
LP & 132.05 &  12.65 & 2.67 & 32.02  \\
Uniform & 2787 & 12.73 & 249.28 &  34.00 \\
\hline
SAMP SC & 143.00 &  53.65 & 72 & 32.00  \\
SAMP KT & 243.00 & 22.74 & 72 & 32.00 \\
\hline
\hline
\end{tabular}

$\quad\quad$

\begin{tabular}{l|l|l|l|l}
\multicolumn{5}{c}{$\MeEt$ in deterministic schedulers}\\
algorithm  &   GN-U & GN-P & GN-Z  &  Clos \\
\hline
\hline
SC     & 72.00   &  43.41 & 48.17 & 16.50  \\
KT     & 122.00  &  11.97 & 4.28  & 16.50  \\
RT CON & 162.00  &  20.15 & 5.31  & 40.61 \\
RT LP  & 173.90  &  18.92 & 2.91  & 40.53 \\
\hline
RT-S CON & 92.50  &  22.15 & 4.90 &    \\
RT-S LP  & 71.50  &  22.16 & 1.90 &    \\
\hline
\hline
\end{tabular}
}

\vspace{0.3in}

\mbox{
\begin{tabular}{l|l|l|l|l}
\multicolumn{5}{c}{$\EtMe$ in deterministic schedulers}\\
algorithm  &   GN-U & GN-P & GN-Z  &  Clos \\
\hline
\hline
SC     & 142.99 &  54.02 & 50.80 & 32.00  \\
KT     & 234.30 &  34.02 & 4.54  & 32.00  \\
RT CON & 345.85 &  55.26 & 6.12  & 147.71 \\
RT LP  & 531.12 &  65.31 & 7.05  & 156.80 \\
\hline
RT-S CON & 182.78  &  42.69 & 6.01 &    \\
RT-S LP  & 142.00  &  43.22 & 3.21 &    \\
\hline
\hline
\end{tabular}

$\quad\quad$

\begin{tabular}{l|l|l|l|l}
\multicolumn{5}{c}{$\MtMe$ in deterministic schedulers}\\
algorithm  &   GN-U & GN-P & GN-Z  &  Clos \\
\hline
\hline
SC     & 143.00  &  95.02 & 113.00 & 32.00  \\
KT     & 243.00  &  35.65 & 9.00   & 32.00  \\
RT CON & 468.00  &  85.72 & 24.00  & 257.00 \\
RT LP  & 833.00  &  97.79 & 13.95 & 225.00 \\
\hline
RT-S CON & 184.00  &  50.00 & 14.00 &    \\
RT-S LP  & 142.00  &  54.00 & 5.00  &    \\
\hline
\hline
\end{tabular}
}
\caption{\MAXe\ objectives.  Table shows expected time with memoryless
schedules (same for all \MAXe\ objectives) and $\MeEt \leq \EtMe  \leq
\MeMt$ on different deterministic schedulers. \label{maxe:table}}
\end{table*}

\begin{figure*}[t]
\centering
\begin{tabular}{lll}
\ifpdf
 \includegraphics[width=0.3\textwidth]{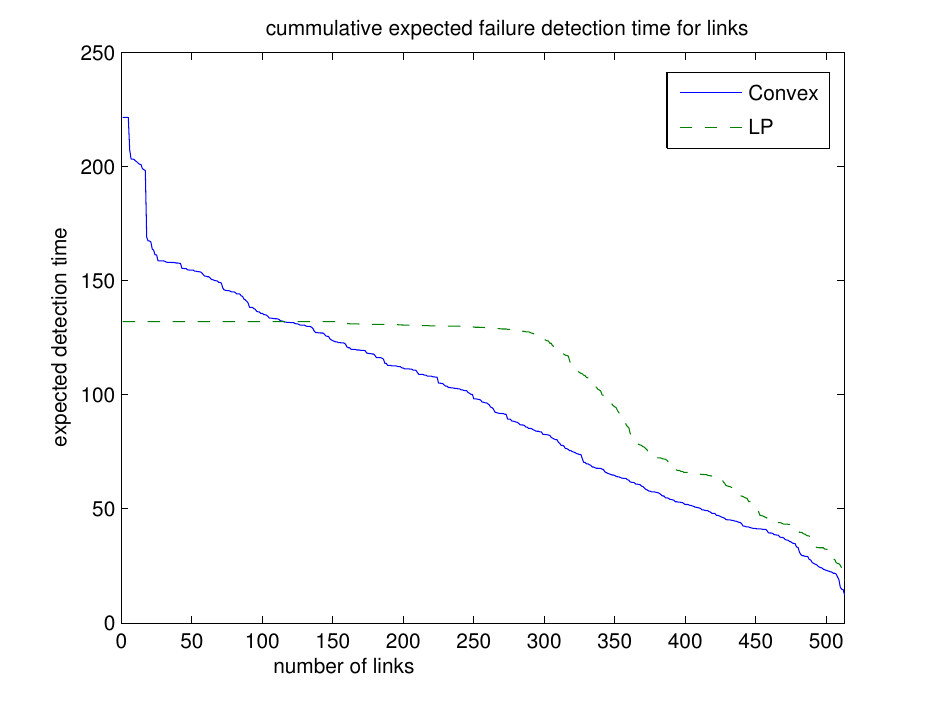}
\else
\epsfig{figure=COnvandLPDistr.eps,width=0.3\textwidth}
\fi
&
\ifpdf
 \includegraphics[width=0.3\textwidth]{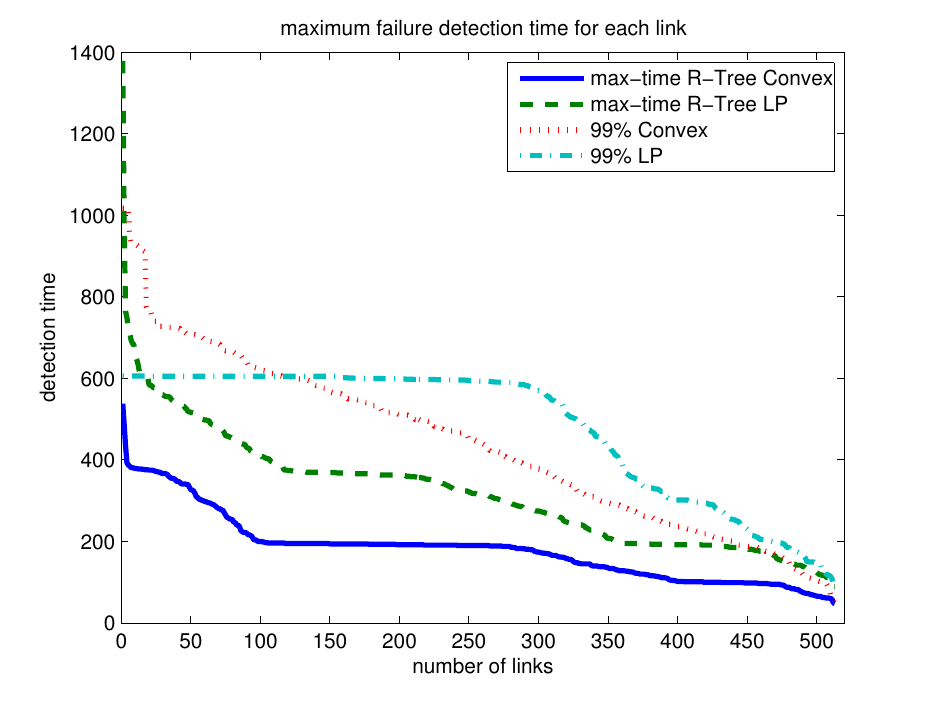}
\else
\epsfig{figure=det-Vs-Rand.eps,width=0.3\textwidth}
\fi
&
\ifpdf
 \includegraphics[width=0.3\textwidth]{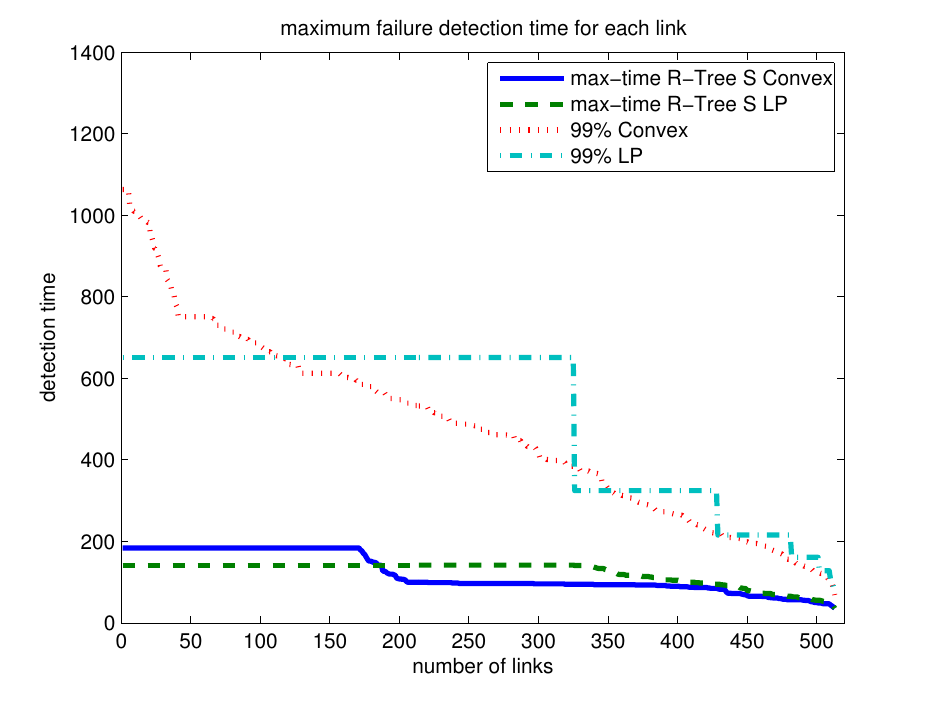}
\else
\epsfig{figure=S-det-Vs-Rand.eps,width=0.3\textwidth}
\fi
\\ 
(A) & (B) & (C) 
\end{tabular}
\caption{(A):  Distribution of time to detect a fault of a link in
  GN-U (the backbone network with uniform priorities).
   (B)-(C):Distribution of time to detect a fault: RT LP and RT CON
  (deterministic) vs. memoryless   over GN-U.  (B): RT LP and RT
  CON  (C): RT-S LP and RT-S CON \label{detVsRTree:fig}\label{sumvmax:fig}}
\end{figure*}

\ignore{
TO DO:

-- Decide on graphs and generate them from real data
    in a format that can be read by the convex program (see   i
    next step)
-- get a convex solver, and find the optimal probabilities ($q_i$s)
-- compute $T_{RAND}(e)$ and $T_{OPT}(e)$
-- implement the deterministic algorithm and compute $T_{DET}(e,t)$.
-- present results.
}

\section{Extension to probabilistic tests} \label{probtest:sec}

A useful extension of our
model allows for
a  probability $\pi_{ei}$ that depends on $i$
and $e$ that a failure to $e$ is found with test $i$.  We assume that
different probes invoking the same or different tests are independent.
Probabilistic tests can model
ECMP (equal cost multi-paths) and transient (inconsistent) failures:
Transient  failures are modeled by
a fixed probability $\pi_{ei}\in (0,1]$ of packet loss.
Tests under ECMP  are modeled by
$\test_i$ being a unit flow between the origin and destination that
defines a probability distribution over tests, where the ``flow''
traversing $e$ is $\pi_{ei}$.

  With probabilistic tests, we may as well use stochastic schedules,
  in particular, memoryless schedules, which also offer strong
  guarantees on the variance of detection times.  Our models and
results for memoryless schedules  have straightforward extensions
to probabilistic tests.
The convex program for opt$_M$-\SUMe\ can be modified to incorporate probabilistic
tests if we replace in \eqref{basicconvex}
$\sum_{i|e\in \test_i}  q_i$  by $\sum_i \pi_{ei}  q_i$.
The LP for opt$_M$-\MAXe\ can be modified by replacing  in
\eqref{basicLP} for each element $e$
$\sum_{i|e\in \test_i} q_i$ by $\sum_{i} \pi_{ei} q_i$.

\section{Related work} \label{sec:related}

   This basic formulation of failure detection via probes applies in multiple
  network scales, from backbone networks to data
  centers~\cite{ZhengCao:IEEEtComp2012,NguyenTTD:infocom09}.
A recent application is
testing of all forwarding rules in a software-defined
network~\cite{ZKVM:conext2012}.
 Beyond the detection of network failures,  the fundamental optimization problems we study model classic and
  emerging resource replication  and capacity  allocation problems.

Previous considerations of the detection problem for network failures
focused on \MAXe\ objective when all elements
 have equal importance (uniform priorities) \cite{ZhengCao:IEEEtComp2012,NguyenTTD:infocom09,ZKVM:conext2012}.
In this particular case, deterministic scheduling is equivalent to finding a
minimum size set of tests which covers all elements, which is
the classic set covering problem.
The optimal memoryless schedule is a solution of a simplified LP,
which computes an optimal fractional cover.
 In practice, however, some elements are much more
critical than others, and the uniform modeling does not capture that.
Ideally, we would like to specify different detection-time targets for
failures which depend on the criticality of the element.
A set cover based
deterministic schedule, however, may
perform poorly when elements have different priorities and there was
no efficient algorithm for constructing good deterministic schedules.
Moreover, the \SUMe\ objectives, which were not previously considered for network
failure detection application, constitute a
natural global objective for overall performance, for example, when
elements have associated fail probability, \SUMe\ minimization
corresponds to minimizing expected failure detection time.

 The special case of {\em singletons} (each test contains a  single
 element) received considerable attention and models
 several important problems.   The \SUMe\ objective on memoryless
 schedules is the subject of Kleinrock's well known ``square root law''~\cite{kleinrock76queueing}.
Scheduling for Teletext~\cite{AmmarWong:TOC1987} and broadcast disks~\cite{AAFZ:sigmod1995}, can be formulated as
deterministic scheduling of singletons. Both $\EeEt$ and $\MeMt$
objectives were considered.
  Our Kuhn-Tucker 
scheduler for \SUMe\ generalizes a classic algorithm for
singletons~\cite{VaidyaHameed:mobicom1997,BBNS:periodic2002}
which has a factor 2 approximation for the $\EeEt$~\cite{BBNS:periodic2002}.
Bar-Noy et al.  \cite{BBNS:periodic2002,BDP:periodic2004}
established a  gap $\leq 2$ between the optimal
deterministic and memoryless schedules, this is in contrast to the
difficulty of general subset tests, where we show that gaps can be asymptotic.
Interestingly, however, even for singletons, $\MeMt$
optimal deterministic scheduling is NP hard \cite{BBNS:periodic2002}.
Several approximation algorithms were proposed for deterministic scheduling~\cite{ksy:stoc2000,BBNS:periodic2002,BDP:periodic2004}.
In particular, Bar-Noy et al.  \cite{BBNS:periodic2002,BDP:periodic2004}
proposed tree-schedules, which are an ingredient in our R-Tree
schedule constructions,  as a representation of deterministic
schedules.
  Memoryless schedules with respect to the
\SUMe\ objective modeled  replication or distribution
of copies of resources geared to optimize the success probabilities or
search times in unstructured p2p networks~\cite{CohenShenker02}.
 Our convex program formulation extends  the solution to a
 natural situation where each test (resource) is applicable to
 multiple elements (requests).

Lastly, our focus here is continuous testing,
which is performed as a background process, but it is also natural to consider
{\em one-time} testing, where a schedule is designed to
be executed once \cite{FLT:approx02,cfk:soda03}. In \cite{onetimesubset:2013} we study the
relation of one-time and continuous testing.


\section*{Conclusion}
We study the fundamental problem of continuous
testing using subset tests.  Our study is comprehensive and unifies models and algorithms.
We reveal the relations between different
objectives and between stochastic and deterministic schedules and
propose efficient scheduling algorithms with provable performance
guarantees. 
For the important application of probe scheduling for silent failure
detection, we conduct
simulations of our algorithms on realistic networks and demonstrate their
effectiveness in varied scenarios. 
Beyond silent failure detection,  we believe the optimization problems we address
and our scheduling algorithms  will find applications in other
resource allocation domains.

\section*{References}
\bibliographystyle{elsarticle-num}
\bibliography{blackholes}

\appendix

\section{R-Tree schedules analysis}\label{rtreedetails:sec}


\subsection{Tree schedules for Singleton tests}
For a given instance,
the best  D2M we can hope for is when the
 deterministic scheduler is able to perform each test in precise
  intervals of $1/q_i$, which results, for singletons instances,  in maximum probing interval of
  $1/q_i$.   Tree schedules achieve this when $q_i=2^{-L_i}$ for all
  $i$.    A deterministic tree schedule for singletons has D2M that is
  at most $2$, and therefore, for all our objectives, the
 D2M gap is at most $2$.

The $\MeMt$ D2M gap and the \EMP\ D2M gap, however,  are {\em exactly}
$2$.
Consider an instance with two elements
one with priority $p_1=1-\epsilon$ and the other with priority
$p_2 = \epsilon$.

Consider $\MeMt$.   The optimal memoryless schedule
\eqref{basicLP} has
$q_1 = 1-\epsilon$ and $q_2=\epsilon$  and $\max_e p_e \Mt[e]= 1$.
Whenever there are at least two elements with positive priorities,
any deterministic scheduler has $\Mt[e]\geq 2$ for all elements.
Therefore, the $\MeMt$
of any deterministic schedule
 is at least $2$ and the D2M is at least $2$.

Consider \EMP.
The optimal memoryless schedule \eqref{basicconvex}
has
$q_1=\frac{\sqrt{1-\epsilon}}{\sqrt{1-\epsilon}+\sqrt{\epsilon}}$
 and $q_2 = \frac{\sqrt{\epsilon}}{\sqrt{1-\epsilon}+\sqrt{\epsilon}}$
and the $\EeMt= p_1/q_1 + p_2/q_2 = (\sqrt{1-\epsilon}+\sqrt{\epsilon})^2 \approx 1$.
A deterministic schedule has
$\Mt[e]\geq 2$ for both elements and thus  $\EeMt = p_1 \Mt[1] +
p_2 \Mt[2] =2$. It follows that the
\EMP\ D2M ration is $\ge 2 -\epsilon$  for any small $\epsilon > 0$.

 Several deterministic schedules for
singletons with ratio at most $2$ (and better than $2$ when possible  for
the particular instance, in particular when priorities are small)
were previously proposed \cite{BBNS:periodic2002,BDP:periodic2004}.
Tree schedules are of interest to us here because they can be
``properly'' randomized to yield good performance
in our treatment of general instances.





\subsection{ R-Tree schedules for subset tests}  \label{rtreeanalysi}

For a single element $e$, we
 analyze the
expected
 (over our randomized
construction of a deterministic tree schedule) maximum probe interval in
the deterministic schedule.  We show
\begin{lemma} \label{onelink}
The expected maximum is $\Theta(\log \nTofe_e)/Q_e$, where
$\nTofe_e=\{ i \mid e\in \test_i\}$ . I.e., for any element $e$,
$$E_{alg}[ \max_t \text{T}(e,t)]\leq c\log (\nTofe_e)/Q_e\ ,$$ where $\text{T}(e,t)$
is the elapsed time from time $t$ until  $e$ is probed.
\end{lemma}
\begin{proof}
Given a level $N$ schedule,
we say that a subinterval of $[0,2^N-1]$ is {\em hit by a test} if contains
  a leaf  of the test.
 We say it is hit by an element $e$ if it is hit by at least one test containing
 the element.

   Consider a  particular element $e$.  We now look only at the tests
   which include the element.  To simplify notation, let $q_i$, $i\in [\nTofe_e]$ be
the frequencies of these tests, let $Q=\sum q_i$, and $q_{\max}=\max_i q_i$.

   We consider the schedule for some level 
$$N \in [\log_2
   (\frac{1}{q_{max}}),\max_i L_i]\ .$$  We will make a precise choice of $N$
   later on.

Note that any interval of size $\geq 1/q_{\max}$ must be hit by the test
with maximum frequency.    We are now looking to bound the
distribution of the  size of the largest interval that is
not hit.

Consider now a subinterval $\subset [0,2^N-1]$ of size $D < 1/q_{\max}$.
We can assume that $D=2^j$ for some $j$ and the interval left endpoint
is an integral multiple of $D$.

  We  upper bound the probability that the interval it is not hit by $e$.
 The probability that it is not ``hit''
by a test with frequency $q_i$ is $q_i D$.  These probabilities of
not hitting the interval by different tests are negatively
correlated:  conditioned on some of the tests not hitting the
interval, it only makes it more likely that other tests do hit the
interval   -- hence,  the probability that the interval is not hit
by any test is at most  the product  $\prod_i (1-q_i D)$, which in
turn is bounded from above by $\prod_i (1-q_i D) \leq \exp(- \sum
q_i D)= \exp(-Q D)$.

  We now upper bound the probability that there exists at least one
 subinterval of size $D=2^{j}$ and left endpoint that is an
 integral multiple of $D$, that is not hit by any test.  We do a union
 bound on $2^N/D$ intervals of this property and this probability is
 at most

\begin{equation} \frac{2^N}{D} \exp(-Q D) \ .\label{unionD}
\end{equation}

 Note that if using $D=\frac{x}{2}$, this upper bounds the probability that there exists an
 interval of size $x$ that is not hit (without restrictions on
 endpoints).  This probability, in terms of $x$, is
\begin{equation} \frac{2^{N+1}}{x} \exp(- Q x/2) \label{gapx}
\end{equation}



   We now restrict our attention to a subset $S$ of the tests which
   satisfy $q_i \geq \frac{Q}{2\nTofe_e}$.   We have $Q_S \equiv \sum_{i\in
     S} q_i  \geq Q/2$.
 We now look only at the tests in $S$. since this is a subset of the tests
that include $e$, it is sufficient to bound the expectation of the
largest open interval with respect to these tests. Since the highest
level in $S$ is  $N=\lceil \log_2 (2 \nTofe_e/Q) \rceil \leq 1+
\log_2 (\nTofe_e/Q) $, we can look at the level $N$ schedule.  We
substitute this $N$ and $Q_S \geq Q/2$ in \eqref{gapx} we obtain
that the probability of an empty interval of size $x$ is
\begin{equation} \label{upp}
\frac{8 \nTofe_e}{xQ} \exp(-xQ/4)\ .
\end{equation}

  For $x=8\ln \nTofe_e/Q$ in \eqref{upp}, we obtain a bound of
$1/(\nTofe_e\ln \nTofe_e ) \leq 1/2$ (for $\nTofe_e\geq
2$, $\nTofe_e=1$ is already covered as $q_{max}$).

We can now obtain an upper bound on the expectation of the maximum
empty interval by
summing over positive integers $i$,   the product of interval size $(i+1)x$
  and an upper bound on the probability of an empty interval of at
  least  size $ix$, for positive integer $i$,
we obtain that the expectation is  $O(x)= (1/Q) O(\ln \nTofe_e)$.
\end{proof}

\smallskip
\noindent
{\bf Proof of Theorem \ref{thm:EMP}}
\begin{proof}
We start with  frequencies $\boldsymbol{q}$ and build a deterministic
tree schedule using our randomized construction.
We show that the expected \EMP\ of the deterministic schedule that we
obtain is at most $\Theta(\ln \nTofe)$ times the \EMP\ of the memoryless
schedule for $\boldsymbol{q}$. To obtain our claim, we take $\boldsymbol{q}$ to be the optimum of
\eqref{basicconvex}.

We apply Lemma~\ref{onelink}.
The lemma shows that for each element $e$ we have $E_{alg}[ \max_t
\text{T}(e,t)]\leq c\log (\nTofe_e)/Q_e$. Now we take a weighted sum over elements using
$\boldsymbol{p}$. We get that,
\[
E_{e\sim p_e} E_{alg}[ \max_t T(e,t)]\leq \sum_e p_e\frac{c\log
(\nTofe_e)}{Q_e}
\]
This is equivalent to,
\begin{align*}
E_{alg} E_{e\sim p_e} [ \max_t \text{T}(e,t)] &\leq \sum_e c\log
(\nTofe_e)\frac{p_e}{Q_e} \\ &\leq c\log (\nTofe_{\max})\sum_e
\frac{p_e}{Q_e}\ .
\end{align*}
This implies that  with probability at
least $1/2$ (over the coin flips of the algorithm)
we get a deterministic schedule whose \EMP\ is $2c\log (\nTofe_{\max})\sum_e
\frac{p_e}{Q_e}$. It follows that the \EMP\ D2M ratio is at most $2c\log (\nTofe_{\max})$.




\EC{ Make a pass over the \EEP\ claim.}
 We now show that the \EEP\ D2M ratio of a random tree schedule is constant with constant probability.
Using the same reasoning as in the proof above for
\EMP\ it suffices to show that for each $e$,
$E_{alg} E_{t}[T(e,t)] \leq c /Q_e$.

Fixing $e$ and an arbitrary  time $t$, as in the proof of Lemma \ref{onelink},
we can easily derive that
$\Pr[T(e,t)\geq D]\le \exp{-Q_eD}$.
In particular we get that
$\Pr[T(e,t)\geq i/Q_e]\le \exp({-i})$. So the fraction of times $t$
in which $T(e,t) \ge i/Q_e$ is at most $\exp({-i})$.
It follows  that
$$E_{alg} E_{t}[T(e,t)] \le  (2/Q_e) \sum_i  \exp(-i) \le c /Q_e$$ for
some constant $c$.
\end{proof}

\smallskip
\noindent
{\bf Proof of Theorem \ref{lemma:treememt}}
\begin{proof} We use \eqref{upp} in the proof of  Lemma~\ref{onelink}.
For an element $e$, the probability of an empty interval of size at
least $x$ is at most
$\frac{8 \nTofe_e}{xQ} \exp(-xQ/4 )$.
Using $x\equiv D_e=8 (\ln \nElements + \ln \nTofe_e)/Q$ we
  obtain
that there is an  interval empty of tests for $e$ of length at least $D_e$
with probability at most $1/n^2$.

By the probability union bound over the elements we get that
 the probability that for all $e$ there is no empty interval
of length more than $D_e$ is at least $1-1/n$.
\end{proof}

\section{Deferred Proofs} \label{deferredp}


\subsection{Proof of Lemma~\ref{allstocequal:lemma}} \label{allstocequal:proof}

The proof of Lemma \ref{allstocequal:lemma} will follow from two
claims. The first claim shows that given a stochastic schedule we
can find a distribution over test sequences of length $N$, such that
the performance of the schedule that repeatedly samples its next $N$
tests from this distribution approaches the performance of the
stochastic schedule we started out with as $N$ approaches infinity.

The second claim shows that given a schedule which is defined, as
above, via a distribution over test sequences of length $N$, we can
define a schedule with the same performance, such that for any fixed
item $e$, the detection time $T(e,t)$ is the same for all times $t$.

We use the following definition.
 A stochastic {\em $N$-test
schedule} $\bsigma_N$ is defined via a distribution $D$ over test
sequences of length $N$,  and it repeatedly samples $D$ to generate
its next $N$ tests.

\begin{claim}
Given a stochastic schedule $\bsigma$, for any $\epsilon>0$ there
exists $N_\epsilon$, such that
 for any $N\ge N_\epsilon$ there is an $N$-test schedule $\bsigma_N$ such that for every $e$ we have
 \[
 \Et[e| \bsigma_N] \leq (1+\epsilon)  \Et[e| \bsigma] \ .
 \]
\end{claim}

\begin{proof}
The next $N$ tests of $\sigma_N$ are obtained  by drawing  a prefix
of $N$ tests from $\sigma$. We will collect constraints on the minimum
size of $N_\epsilon$  and eventually pick $N_\epsilon$ to be
large enough to satisfy all these constraints.

Now, since the schedule $\bsigma_N$ samples sequences of length $N$
repeatedly from the same distribution, we can consider the time
modulus $N$, hence,
\begin{eqnarray*}
\Et[e| \bsigma_N]&=& \lim_{h \rightarrow \infty}
\frac{1}{h}\sum_{t=1}^{h} \bT(e,h|\bsigma_N) = \frac{1}{N}\sum_{t=1}^N \bT(e,t| \bsigma_N) \ .
\end{eqnarray*}
So we have to show that for sufficiently large $N$
$$
\frac{1}{N}\sum_{t=1}^N \bT(e,t| \bsigma_N) \le \Et[e|\bsigma]
(1+\epsilon) \ .
$$

Denote by $\bT^*(e,t | \bsigma)$ the random variable of the cover
time of $e$ at time $t$ by the schedule $\sigma$, so
$\E[\bT^*(e,t|\bsigma)]=\bT(e,t|\bsigma)$.

By the definition of $\bsigma_N$ we have that for any $t$, $1\le
t\le N$,
\begin{eqnarray} \label{eq:wrapeq}
\text{T}(e,t |{\bsigma_N} ) &\leq & \text{T}(e,t|{\bsigma} )+  \Pr[\text{T}^*(e,t|\bsigma)> N-t] \text{T}(e,1|{\bsigma_N}) \ .
\end{eqnarray}

From Markov inequality, applied to the random variable $\bT^*(e,1 |
\bsigma)$, we get that $\Pr[\text{T}^*(e,1|\bsigma) \ge N] \leq
\frac{\text{T}(e,1|\bsigma)}{N}$. Picking $N_\epsilon \geq \max_e
\bT(e,1|\bsigma)/\epsilon$ we have that $\Pr[\text{T}^*(e,1|\bsigma)
\ge N] \leq \epsilon$ for all items $e$. Substituting this and $t=1$
in Equation (\ref{eq:wrapeq}) we get that
\begin{equation*}
T(e,1|\bsigma_N) \leq \bT(e,1|\bsigma)+\epsilon \bT(e,1|\bsigma_N) \
.
\end{equation*}
which implies that
\begin{equation} \label{t1sigmarel1}
\bT(e,1|\bsigma_N) \leq \frac{\bT(e,1|\bsigma)}{(1-\epsilon)} \ .
\end{equation}
Substituting Equation (\ref{t1sigmarel1}) back into
(\ref{eq:wrapeq}) we get
\begin{eqnarray} \label{eq:wrap1}
\text{T}(e,t |{\bsigma_N} ) &\leq & \text{T}(e,t|{\bsigma} )+ \Pr[\text{T}^*(e,t|\bsigma)> N-t]
\frac{\bT(e,1|\bsigma)}{(1-\epsilon)} \ .
\end{eqnarray}
 Markov inequality, for any time $t$, gives
\begin{equation}  \label{sumprob}
\Pr[\text{T}^*(e,t)> N-t]\leq \min\{1,
\frac{\text{T}(e,t|\bsigma)}{N-t+1} \}\ .
\end{equation}
Substituting this in (\ref{eq:wrap1}) we get
\begin{eqnarray} \label{eq:wrap2}
\text{T}(e,t |{\bsigma_N} ) & \leq & \text{T}(e,t|{\bsigma} )+  \min\{1, \frac{\text{T}(e,t|\bsigma)}{N-t+1} \}
\frac{\bT(e,1|\bsigma)}{(1-\epsilon)}  \ .
\end{eqnarray}
We now sum  (\ref{eq:wrap2}) over all $1\le t \le N$
\begin{eqnarray} \label{eq:wrap2}
\sum_{t=1}^N\text{T}(e,t |{\bsigma_N} ) & \leq &\sum_{t=1}^N \text{T}(e,t|{\bsigma} )+  \frac{\bT(e,1|\bsigma)}{(1-\epsilon)} \sum_{t=1}^N \min\{1,
\frac{\text{T}(e,t|\bsigma)}{N-t+1} \} \ .
\end{eqnarray}

Our goal now is to bound the second term on the right hand side of
(\ref{eq:wrap2}). Since we only consider valid schedules, for each
$e$, there must be $N_{e,\epsilon}$ so that for all $h\geq
N_{e,\epsilon}$,
\begin{equation}\label{limitN}
\frac{1}{h}\sum_{t=1}^{h} \text{T}(e,t | \boldsymbol{\sigma}) \leq
\text{\EP}[e |\boldsymbol{\sigma} ](1+\epsilon)\ .
\end{equation}
We will select $N_\epsilon \geq \max_e N_{e,\epsilon}$ so
(\ref{limitN}) holds for any $h=N \ge N_\epsilon$.

It follows that to upper bound  $\sum_{t=1}^N \min\{1,
\frac{\text{T}(e,t|\bsigma)}{N-t+1} \}$ we can consider the
following optimization problem:
\begin{eqnarray*}
\max \,\sum_{t=1}^N \min\{1,\frac{x_t}{N-t+1}\}  && \text{s.t.}\,  \sum_{t=1}^N x_t\leq B
\end{eqnarray*}
where in our setting $x_t=T(e,t)$ and $B=(1+\epsilon)N\Et[e |\sigma
]$. We substitute $y_t = x_{N-t+1}$ and the optimization problem
simplifies to
\begin{eqnarray*}
\max \, \, \sum_{t=1}^N \min\{1,\frac{y_t}{t}\} && 
\text{s.t.} \, \, \sum_{t=1}^N y_t\leq B
\end{eqnarray*}

The solution to the optimization is to set $y_t=t$ for $t\in [1,z]$
for the largest $z$ such that $\sum_{j=1}^z j = (1+z)z/2 \le B$,
$y_{z+1} = B - (1+z)z/2$ and $y_t=0$ for $t\ge z+2$. We get that $z
\le \sqrt{2B}$ and $\sum_{t=1}^N \min\{1,\frac{y_t}{t}\} \le
\sqrt{2B} + 1$.

Substituting this bound back in (\ref{eq:wrap2}) we get that
\begin{eqnarray}
\sum_{t=1}^{N} \text{T}(e,t|{\boldsymbol{\sigma}_N}) 
&\leq& \sum_{t=1}^{N} \text{T}(e,t |{\boldsymbol{\sigma}}) +
\frac{\text{T}(e,1|{\boldsymbol{\sigma}})}{1-\epsilon}
\bigg(1+\sqrt{2N(1+\epsilon)
  \Et[e|\bsigma]} \bigg)\nonumber \\
&\leq&   (1+\epsilon)
N {\Et}[e | {\boldsymbol{\sigma}}] +
\text{T}(e,1|{\boldsymbol{\sigma}}) 4 \sqrt{N \Et[e| \bsigma]} \label{st4} \\
&=& N \Et[e|\bsigma] \bigg(1+\epsilon +
\frac{4\text{T}(e,1|{\boldsymbol{\sigma}})}{\sqrt{N \Et[e| \bsigma]
}}\bigg) \nonumber
\end{eqnarray}
where inequality  (\ref{st4}) is by substituting (\ref{limitN}) and
assuming that $\epsilon <0.5$.

We will now set $N_\epsilon$ appropriately. First we need
$N_\epsilon \geq \max_e N_{e,\epsilon}$. Second, we need that
$N_\epsilon \geq \max_e \bT(e,1)/\epsilon$. Third,  $N_\epsilon$
should be large enough so   that
$\frac{4\text{T}(e,1|{\boldsymbol{\sigma}})}{\sqrt{N_\epsilon \Et[e|
\bsigma] }} \leq \epsilon$.
We get that
\[
 \Et[e| \bsigma_N] = \frac{1}{N}\sum_{t=1}^{N} \bT(e,t|\bsigma_N) \leq \Et[e|\bsigma] (1+2\epsilon)
\]
from which the proof follows by using $\epsilon/2$ rather than
$\epsilon$ in our constraint on $N_\epsilon$ specified above.
\end{proof}

A stochastic shifted  $N$-test schedule $S(\bsigma_N)$ is defined with
respect to
a stochastic $N$-test schedule $\bsigma_N$ as follows. It samples
uniformly a random $i\in[1,N]$ and a sequence $x$ from $\sigma_N$
(recall that $x$ is an infinite sequence of tests composed from
blocks of $N$ tests) and starts from test $i$ in $x$.

\begin{claim}
Given $\sigma_N$, for any $t$ and $e$ we have,
\begin{equation}
\bT(e,t|S(\bsigma_N))= \Et[e| \bsigma_N] \ .
\end{equation}
\end{claim}

\begin{proof}
We first will show that $\bT(e,t|S(\bsigma_N))$ is independent of
$t$ and then show that it equals $\Et[e| \bsigma_N]$. We can see
that it is independent of $t$ by the definition of $S(\bsigma_N)$
from which we get
\begin{eqnarray*}
\bT(e,t|S(\bsigma_N)) &=& \frac{1}{N}\sum_{i=1}^N \bT(e,t+i| \bsigma_N) 
= \frac{1}{N}\sum_{i=1}^N \bT(e,(t+i) \bmod N | \bsigma_N) 
= \frac{1}{N}\sum_{i=1}^N \bT(e,i| \bsigma_N) \ .
\end{eqnarray*}
Now, since the schedule $\sigma_N$ samples sequences of length $N$, we can consider
the time modulus $N$, hence,
\begin{eqnarray*}
\Et[e| \sigma_N]&=& \lim_{h \rightarrow \infty}
\frac{1}{h}\sum_{t=1}^{h} \bT(e,h|\bsigma_N) = \frac{1}{N}\sum_{i=1}^N \bT(e,i| \bsigma_N)
\end{eqnarray*}
\end{proof}

\subsection{Proof of Lemma~\ref{NPhard}}  \label{nphardproof:sec}
\begin{proof}   We obtain a
 scheduling instance using the same set of elements and subsets
 (tests) as the X3C instance.
We use a uniform $\boldsymbol{p}$ over elements with $p_e=1/(3k)$ for
\SUMe\ objectives and $p_e=1$ for \MAXe\ objectives.


We first consider deterministic schedules.
  From an exact cover, we define a deterministic schedule by
  cycling through the same permutation of the cover.
  The deterministic schedule has
$\Mt[e]=k$ and  \EP$[e]=(k+1)/2$ for all elements $e$.  The
maximum $\max_e T(e,t)$ at any time $t$ is $k$ and the average is
$(k+1)/2$.  Therefore, the schedule has \WMP, \MWP, and \EMP\ equal to
$k$ and \WEP, \EEP, and \MEP\ equal to $(k+1)/2$.

Consider an arbitrary deterministic schedule and
 time $t$.  We must have $\max_e T(e,t) \geq k$,
since at most $3i$ elements can be covered in $i$ probes, so to cover
all $3k$ elements we need at least $k$ probes.  We have
equality
 if and only if the
sequence of $k$ probes following $t$ constitutes a cover.  A cover
of size $k$ must be an exact cover.  Therefore \MWP$=k$ implies
exact cover of size $k$.

Similarly, we claim  that on any schedule, $(1/k)\sum_e T(e,t)
\geq (k+1)/2$.
This is because $\sum_e T(e,t)=\sum_e m_e$, where $m_e$ is the
smallest $d$ such that $e\in \sigma_{t+d}$.  Since there can be at
most 3 elements of each value of $m_e\geq 1$, we have that
$\sum_e T(e,t) \geq 3 \sum_{d=1}^{k} d = 3k(k+1)/2$ and our claim follows.
\ignore{
This is because
\begin{equation}  \label{s2}
\sum_e T(e,t)=\sum_e\sum_{d\geq 1}
\pi_{t+d,e} d\ ,
\end{equation}
 where $\pi_{t+d,e}$ is the probability that
$\sigma_{t+d}=e$ given that $e\not\in
\sigma_{t+1},\dots,\sigma_{t+d-1}$.
Since $\forall e, \sum_{d\geq 1} \pi_{t+d,e} =1$, the sum \eqref{s2}
is minimized when $\pi$ are concentrated in the lower $d$'s, that is
$\sum_e \pi_{t+d,e}=3$ for $d=1,\ldots,k$.
Thus we obtain that
$\sum_e T(e,t) \geq 3 \sum_{d=1}^{k} d = 3k(k+1)/2$.
}
Moreover, equality holds only if the sequence of $k$ probes from $t$ on is
an exact cover.  Therefore  \MEP$=(k+1)/2$ implies exact cover of size $k$.

  Consider an arbitrary deterministic schedule and let $q_e$ be the
  average probing frequency of element $e$ (recall that we only
  consider valid schedules, where $q_e$ is well defined).
We have $\Mt[e] \geq 1/q_e$ and \EP$[e]\geq (1+1/q_e)/2$.
Moreover,
equality can hold only when $1/q_e$ is integral and probes are
evenly spaced every $1/q_e$ probes except for vanishingly small
fraction of times.  For the X3C instance we have $\sum_e q_e=3$, and
from convexity, $\sum_e 1/q_e$ or $\max_e 1/q_e$ are minimized only when all $q_e$ are
equal to $1/k$.  This means that the \WMP\ and \EMP\ can be
equal to $k$ or the
 \WEP\ and \EEP\ are equal to $(k+1)/2$ equal to $(k+1)/2$
only if each element is probed every $k$ probes (except vanishingly
small) number of times.  This means that most sequences of $k$
consecutive probes constitute an exact cover.


 We now consider stochastic schedules.
 From an exact cover, we define a stochastic schedule by a uniform distribution ($1/k$) on each of the $k$ shifts
of the same permutation of the cover.  On this schedule, all our
objectives have value $(k+1)/2$.   It remains to show that for each of
the objectives, a schedule with time $(k+1)/2$ implies an exact cover.

  Observe that with our choice of weighting, on any schedule opt-\MAXe
  $
  \geq $ opt-\SUMe.
  Therefore
if the stochastic optimum of either the \SUMe\ or \MAXe\ objectives is $(k+1)/2$, then opt-\SUMe\ is also
$(k+1)/2$ which implies, from \eqref{EEPSD:claim}, that
opt$_D$-\EEP$=(k+1)/2$,  which implies
 exact cover.
\end{proof}

\subsection{Proof of Lemma~\ref{SvsDlemma}} \label{svsdlemmaproof:sec}
\begin{proof}
We first establish \eqref{EEPSD:claim}.
We show that given a stochastic schedule $\bsigma$ and
$\delta>0$, we can construct a deterministic
schedule $\bsigma_D$, such that $\EeEt[\bsigma_D] \leq (1+\delta)
\EeEt[\bsigma]$.
The main difficulty which makes this proof more technical stems from existence of valid stochastic schedules with
deterministic instantiations which are not valid (limits and
frequencies are not well defined). Therefore, we can not simply assume
a positive probability of a (valid) deterministic schedule with an average
cost that is close to that of $\bsigma$.

Our construction consists of several steps.  We first show that
there is a deterministic testing sequence\footnote{We use the term
  sequence rather than schedule because the sequence may not be a
  valid schedule.}  $\bsigma'$ so that the
average cost on the first $N$ time steps (for sufficiently large $N$
that depends on $\epsilon$) is within $(1+\epsilon)$ of that of the
stochastic schedule.  We then focus on a sub-sequence of steps
$[t_0,N]$ of size $\Omega(N)$ so that the average property still
holds and in addition, the cost of steps $t_0$  is at most a
constant times the average. We then argue that the maximum interval
between tests of an element on the prefix of $\sigma'$ is bounded by
a value $X=O(\sqrt{N})$.
 Lastly, we obtain $\sigma_D$ as a cyclic schedule
which repeats steps $[t_0,N]$ of $\sigma'$.  We show that the
average cost is within $(1+O(\epsilon))$ from the average cost on
times $[t_0,N]$ of $\sigma'$ which in turn, is within $(1+\epsilon)$
to the average cost of the  original $\bsigma$.

From $\bsigma$ being valid, there must be $N_{EE} >
0$ such that for all $N\geq N_{EE}$
$$\frac{1}{N}\sum_{t=1}^N \sum_e p_e \text{T}(e,t) \leq (1+\epsilon)
\EeEt[\bsigma]\ .$$

 Fix some $N\geq N_{EE}$. We draw a particular  execution of $\bsigma$
obtaining an infinite  deterministic sequence $\sigma'$. From Markov
inequality, with probability at least $1-(1+2\epsilon)/(1+\epsilon)
> 0$,
\begin{equation} \label{detcond}
 \frac{1}{N}\sum_{t=1}^N \sum_e p_e \text{T}(e,t | \sigma') \leq (1+2\epsilon)
\EeEt[\bsigma]\ .
\end{equation}
We therefore assume that we have a sequence $\sigma'$ which
satisfies \eqref{detcond}.

We now focus on a subset $[t_0,N]$ of time steps, where $t_0$ is the
minimum $t$ such that $\sum_e  p_e \text{T}(e,t | \sigma') \leq 10
\EtEe[\bsigma]$. From \eqref{detcond}, assuming $\epsilon \le 1/2$,
it follows that $t_0 \leq 0.2 N$. Let  $N'=N-t_0+1 \geq 0.8 N$ be
the length of the interval $[t_0,N]$. We establish that
\begin{equation} \label{detcondpp}
 \frac{1}{N'}\sum_{t=t_0}^{N} \sum_e p_e \text{T}(e,t | \sigma') \leq (1+2\epsilon)
\EeEt[\bsigma]\ .
\end{equation}
We establish \eqref{detcondpp} using \eqref{detcond}:
\begin{eqnarray*}
\sum_{t=t_0}^{N}
\sum_e p_e \text{T}(e,t | \sigma') 
& = & \sum_{t=1}^{N} \sum_e p_e \text{T}(e,t | \sigma')
- \sum_{t=1}^{t_0-1}\sum_e p_e \text{T}(e,t | \sigma') \\
& \leq&  N  (1+2\epsilon) \EeEt[\bsigma]
- ( t_0-1) 10 \EeEt[\bsigma] \\
&= &N' (1+2\epsilon) \EeEt[\bsigma]\ .
\end{eqnarray*}

We now bound the maximum elapsed times between tests of an element
$e$ in the sequence $\sigma'$  in the time interval $[t_0,N]$.
Consider an interval $[i,i+x_e-1]$  of  $x_e$ time steps, completely
contained in $[t_0,N]$ (that is $i+x_e-1 \le N$) in which element
$e$ is not tested then

\begin{equation} \label{xecond1}
\sum_{t=i}^{i+x_e-1} \text{T}(e,t|\sigma') = \sum_{j=1}^{x_e} j \geq
x_e^2/2\ .
\end{equation}
On the other hand, since $\sigma'$ satisfies \eqref{detcond}, noting
that $i+x_e-1\leq N$, we must have
\begin{equation} \label{xecond2}
p_e \sum_{t=i}^{i+x_e-1} \text{T}(e,t|\sigma') \leq \sum_{t=1}^N
\sum_e p_e \text{T}(e,t|\sigma') \leq N(1+2\epsilon) \EeEt[\bsigma]\
.
\end{equation}
Combining \eqref{xecond1} and \eqref{xecond2}, we obtain that
\begin{equation} \label{xbound}
x_e \leq \sqrt{\frac{2(1+2\epsilon) \EtEe[\bsigma] N}{p_e}} \ .
\end{equation}
Let \begin{equation} \label{Xdef} X = \sqrt{\frac{2(1+2\epsilon)
\EtEe[\bsigma] N}{\min_e
    p_e}}\ ,
\end{equation}
we established that
\begin{equation} \label{claimsigmapp}
\forall e \forall t\in [t_0,N-X+1],\  \text{T}(e,t | \sigma') \leq
X\ .
\end{equation}

 Lastly, we define the deterministic schedule $\bsigma_D$ which cycles through the
 steps $[t_0, N]$ of $\sigma'$.  Since $\bsigma_D$ is
  cyclic, we have
\begin{equation} \label{sigmad}
\EeEt[\bsigma_D]=\frac{1}{N'} \sum_{t=1}^{N'} \sum_e p_e \text{T}(e,t |
\bsigma_D)\ .
\end{equation}
We therefore upper bound the latter by relating it to $\sigma'$.

\begin{eqnarray}
\sum_{t=1}^{N'} \sum_e p_e \text{T}(e,t |
\bsigma_D) 
&\leq&   \sum_{t=t_0}^{N} \sum_e p_e  \text{T}(e,t | \sigma') + X
\sum_e p_e  \text{T}(e,t_0 |
\sigma')  \nonumber \\
&\leq&   \sum_{t=t_0}^{N} \sum_e p_e  \text{T}(e,t |
\sigma')  +  10 X  \EtEe[\bsigma] \label{stt1} \\
&\leq&  N'(1+2\epsilon) \EtEe[\bsigma]  + \epsilon N
\EtEe[\bsigma] \label{stt3} \\
&\leq& N' \EtEe[\bsigma] (1+2\epsilon + \epsilon (N/N')) \nonumber \\
&\leq& N' \EtEe[\bsigma] (1 + 4\epsilon) \label{stt5}\ .
\end{eqnarray}
To verify the first inequality,  we apply \eqref{claimsigmapp}
obtaining that  for $t\leq N-X+1$, $\text{T}(e,t-t_0+1|\bsigma_D) =
\text{T}(e,t|\sigma')$.  For the remaining $X$ time steps that
correspond to $t\in (N-X+1,N]$ of $\sigma'$ ($t\in (N'-X+1, N']$ of
$\bsigma_D$) we have
$$
\text{T}(e,t-t_0+1|\bsigma_D) \leq
\left\{\begin{array}{l}\text{T}(e,t|\sigma'),
  \text{ if } \text{T}(e,t|\sigma') \leq N-t \\
N-t +\text{T}(e,1 | \bsigma_D) \text{ , otherwise.}\end{array}
\right.
$$
\begin{align*}
\leq  \text{T}(e,t|\sigma')
+\text{T}(e,t_0|\bsigma_D) =  \text{T}(e,t|\sigma') +\text{T}(e,t_0|\sigma')\ .
\end{align*}
Inequality \eqref{stt1} follows from our choice of $t_0$. Inequality
\eqref{stt3} holds if we choose
 $$N \ge \frac{200(1+2\epsilon) \EtEe[\bsigma]}{\epsilon^2 \min_e
    p_e} \ ,$$ to guarantee that $10X \le \epsilon N$.
 Lastly, \eqref{stt5} uses $N' \geq 0.8 N$.
Combining \eqref{stt5} with \eqref{sigmad}, we obtain
$\EeEt[\bsigma_D] \leq (1+4\epsilon) \EtEe[\bsigma]$. We conclude
the proof of \eqref{EEPSD:claim} by choosing $\epsilon = \delta/4$.

\medskip

 We now establish the inequalities \eqref{EMP:claim} and \eqref{WMP:claim}.
Given a deterministic schedule $\bsigma$, we define a cyclic
deterministic schedule $\bsigma_C$
which repeats a sequence $\bsigma_C'$ of some length $N$ and which
satisfies
$\forall e,\ \Mt[e|\bsigma_C] \leq \Mt[e|\bsigma]$.
Consider the schedule $\bsigma$ and associate a state with each time
$t$, which is a vector that for each $e$, contains the elapsed number
of steps since a test for $e$ was last invoked.  At $t=1$ we have the
all zeros vector.  When a test $\test$ is invoked, the entries for all
elements in $s$ are reset to $0$ and the entries of all other elements
are incremented by $1$.
From definition, the maximum value for entry $e$ is $\Mt[e]$.
Therefore, there is a finite number of states.
The segment $\bsigma_C'$ is any sequence between two times with the
same state.  It is easy to see that the cyclic schedule $\bsigma_C$
obtained
from $\bsigma_C'$ has the desired property.

\ignore{
Given a deterministic schedule and $\epsilon$, we construct a cyclic
deterministic schedule on which the \EMP\ and \MEP\ are within
$(1+\epsilon)$ of the original deterministic schedule.  We take
$N_\epsilon \gg \text{\EMP}/\epsilon$ and such that the \EMP\ and
\MEP\ on the first $h$ time steps for all
$h\geq N_\epsilon$ are at most $(1+\epsilon/2)$ times that of the original
sequence. We now
look at the ``state'', which is the last time since each edge was
tested and time for next probe.  There is a bound $L$ such that the
state at most time steps has all times at most $L$.
We can now find a subsequence of the schedule that starts and finishes
at the same state and that the objectives over it are at most
$(1+\epsilon)$ times those
of the original schedule.  We turn it into a cyclic schedule.
}

  We now take the deterministic  cyclic schedule $\bsigma_C$ and construct a stochastic schedule  $\boldsymbol{\sigma}'$
  by selecting a start point $i\in N$ uniformly at random, executing
  steps $[i,N]$ of $\bsigma_C'$, and then using $\bsigma_C$.
For each element $e$, we have
 $${\Mt} [e  | {\boldsymbol{\sigma}'}] \leq
 \frac{{\Mt}[e |{\boldsymbol{\sigma}_C} ]+1}{2} \leq \frac{{\Mt}[e
   |{\boldsymbol{\sigma}} ]+1}{2}  \ .$$
By combining,
\begin{eqnarray*}
\text{opt-\EMP}  &\leq& \sum_e p_e {\Mt} [e | {\boldsymbol{\sigma}'} ]
\leq \sum_e \frac{{\Mt}[e | {\boldsymbol{\sigma}}]+1}{2}  =   (\text{\EMP}[\bsigma]+1)/2\ .
\end{eqnarray*}
By taking the infimum of \EMP\ over all deterministic schedules we conclude
the claim.  The argument for \WMP\ is similar.
\end{proof}


\subsection{Proof of Lemma \ref{lowerboundlink}} \label{lowerboundlink:sec}
\begin{proof}

We choose $\nElements$, $\nTofe \geq 1$,  and $\nTests \geq 2\nTofe $ such that
$\nElements={\nTests \choose \nTofe}$, and
construct an instance with $\nElements$ elements and $\nTests$ tests
such that each element is included in exactly $\nTofe$ test and every
subset of $\nTofe$ tests has exactly one common element.  We use a
uniform $\boldsymbol{p}$.

  The instance is symmetric and therefore  in the solution of the
 convex program \eqref{basicconvex} or \eqref{basicLP} all the $\nTests$ tests have equal
rates  $q=1/\nTests$.  The memoryless schedule with this
$\boldsymbol{q}$ optimizes both \SUMe\ and \MAXe\ for  $\boldsymbol{p}$.
For any element and any time, the expected detection time by a memoryless
schedule with $\boldsymbol{q}$ is
$\nTests/\nTofe$.  But for any particular
deterministic schedule and a particular time
 there is an element that requires
$\nTests-\nTofe$ probes (for any sequence of $\nTests-\nTofe$ tests there must be at
least $\nTofe$ tests not included and we take the element in the
intersection of these tests.
 This means that at any time, the  worst-case element detection time
is factor $\frac{\nTests-\nTofe}{\nTests/\nTofe} \geq \nTofe/2  =  \Theta(\ln \nElements / \ln \nTests)$ larger than the
memoryless optimum.

When fixing the number of tests $\nTests$, this is maximized (Sperner's
Theorem)  with $\nTofe=
\nTests/2$ and the \MAXe\ ratio is $\Theta(\nTests)$.
 When fixing the number of elements $\nElements$,  the maximum ratio is
$\arg\max_\nTofe \nElements={2\nTofe \choose \nTofe}$ and we obtain
$\nTofe = \Theta(\ln \nElements)$.

We use the same construction as in Lemma~\ref{lowerboundlink} and take a uniform $\boldsymbol{p}$ over elements.
Lastly, to show \EMP\ of $\Omega(\ln \nTofe)$ consider
a sequence of $\nTests$ probes.  The expectation over elements of the
number of probes that test the element, is at most  $\nTofe$. So at least
half the elements are probed at most $\nTofe$ times.  There are at least
$\nTests/2$ distinct tests.  The expected over $e$ maximum difference between
probes to an element $e$ over a sequence of $\nTests$ is $\Omega(\ln
\nTofe) \nTests/\nTofe$.  This is because every combination of $\ell$
distinct probes corresponds to an element, and thus, for the ``average''
element, the probes can be viewed as randomly placed, making the
expectation of the maximum interval a logarithmic factor larger than
the expectation.

 We now show how the instances can be realized on
 a network.
   We use $\nElements$  pairs of links.  Each pair includes a
 link which corresponds to an element in our instance and a
``dummy'' link.  The pairs are connected on a path of size
$\nElements$.
Each test is an end to end path which traverses one link from each
pair.
Each (real) link is covered by exactly $\nTofe$ paths
and every subset of $\nTofe$ paths has one common (real)  link.  The network is
a path of length $\nElements$ of pairs of parallel links, a real and a dummy
link. Real  links $e$ have  $p_e=1$ and the dummy link have $p_e=+\infty$.
(If we want to work with respect to some \SUMe\ optimum we take
 $p_e\equiv p$ (for some $p<1$) for real links and $p_e=0$ for dummy links).
Each path traverses one link from each pair  and includes $\nTofe$ real
links.
\end{proof}
The Lemma is tight in the sense that it is always possible to get a
schedule with D2M equal to $\nTests$ by cycling over a permutation of
the tests.

\end{document}